\documentclass[journal]{IEEEtran}

\usepackage[utf8]{inputenc}
\usepackage[T1]{fontenc}
\usepackage{cite}
\usepackage[cmex10]{amsmath}
\usepackage{amssymb}
\usepackage{amsthm}
\usepackage{bm}
\usepackage{graphicx}
\usepackage{epstopdf}
\usepackage[caption=false,font=footnotesize]{subfig}
\usepackage{array}
\usepackage{multirow}
\usepackage{makecell}
\usepackage{algorithm}
\usepackage{algpseudocode}
\usepackage{url}
\usepackage[hidelinks]{hyperref}

\allowdisplaybreaks
\hypersetup{
     colorlinks   = true,
     citecolor    = blue,
     linkcolor    = blue,
     urlcolor     = black
}

\newtheorem{theorem}{Theorem}
\newtheorem{lemma}{Lemma}
\newtheorem{proposition}[theorem]{Proposition}

\newtheorem{remark}{Remark}
\newtheorem{assumption}{Assumption}

\newcommand{\etal}{\emph{et al. }}

\usepackage{eso-pic}
\begin{document}

\AddToShipoutPictureBG*{%
  \AtPageUpperLeft{%
    \put(\LenToUnit{0.5\paperwidth},\LenToUnit{-1.2cm}){%
      \makebox[0pt]{\parbox{0.9\paperwidth}{\scriptsize\centering
        \copyright~2026 IEEE. Personal use of this material is permitted. Permission from IEEE must be obtained for all other uses, in any current or future media, including reprinting/republishing this material for advertising or promotional purposes, creating new collective works, for resale or redistribution to servers or lists, or reuse of any copyrighted component of this work in other works.}}}}}

\title{Active STAR-RIS-Aided RSMA IoT Systems: Performance Analysis, Model-Based and Data-Driven Resource Allocation Frameworks}
\author{Ngo~Hoang~Tu,~\IEEEmembership{Member,~IEEE},
        Vo~Nguyen~Quoc~Bao,~\IEEEmembership{Senior Member,~IEEE},\\
        and Tran~Thien~Thanh,~\IEEEmembership{Member,~IEEE}%
    \thanks{This work was supported by the University-level Scientific Research Project at Ho Chi Minh City University of Transport with Code KHT\DJ{}2401/H\DJ{}-\DJ{}HGTVT. (\textit{Corresponding author: Tran Thien Thanh}.)}
    \thanks{Ngo Hoang Tu and Vo Nguyen Quoc Bao are with the Faculty of Information Technology, Van Lang School of Technology, Van Lang University, Ho Chi Minh City 70000, Vietnam (e-mail: tu.nh@vlu.edu.vn, bao.vnq@vlu.edu.vn).}
    \thanks{Tran Thien Thanh is with the Institute of Information Technology and Electrical-Electronics Engineering, Ho Chi Minh City University of Transport, Ho Chi Minh City 70000, Vietnam (e-mail: thanh.tran@ut.edu.vn).}
}
\maketitle

\begin{abstract}
Simultaneously transmitting and reflecting reconfigurable intelligent surfaces (STAR-RISs) have emerged as a promising technology for achieving full-space signal coverage with low hardware complexity and energy consumption.
In parallel, rate-splitting multiple access (RSMA) offers a flexible interference management mechanism that enhances spectral efficiency and user fairness.
Integrating STAR-RIS with RSMA provides strong potential for robust and energy-efficient communications; however, the resulting double-fading cascaded channels and the inherent imbalance between near and far users pose reliability and fairness challenges.
This study presents comprehensive information-theoretic and optimization frameworks for active STAR-RIS-assisted RSMA Internet-of-Things systems.
Closed-form expressions for the outage probability (OP) and ergodic capacity (EC) are derived under both the presence and absence of direct links.
Based on these results, asymptotic OP and EC, diversity order, and array gain are analyzed to characterize system behavior.
In addition, throughput-based and spectral-based energy efficiency metrics are investigated to quantify the tradeoffs between transmission performance and power consumption.
To enhance user fairness, a fairness-oriented resource allocation (RA) framework is developed to jointly optimize power allocation and rate-splitting coefficients using successive convex approximation and block coordinate descent techniques.
To further improve scalability and real-time applicability, we introduce a data-driven RA framework based on a deep multi-output neural
network, a convolutional multi-output neural network, and multi-output extreme gradient boosting (MXGB). This framework approximates the optimization-based solutions with substantially reduced computational complexity.
Numerical results show that the proposed MXGB model achieves the shortest execution time while maintaining performance comparable to ground-truth-based benchmarks. Extensive Monte Carlo simulations validate the analytical results and demonstrate the effectiveness of the proposed frameworks.
\end{abstract}

\begin{IEEEkeywords}
Simultaneously transmitting and reflecting reconfigurable intelligent surface (STAR-RIS),
rate-splitting multiple access (RSMA),
outage probability,
ergodic capacity,
outage fairness optimization,
capacity fairness optimization,
and machine learning.
\end{IEEEkeywords}

\section{Introduction}
\label{Sect:Intro}

\subsection{Research Context}
The rapid growth of data-intensive applications and heterogeneous service requirements in next-generation networks has driven the need for Internet of Things (IoT) paradigms that jointly enhance spectral efficiency, reliability, and energy efficiency (EE) \cite{dogra2020survey}.
Reconfigurable intelligent surfaces (RISs) reshape the wireless propagation environment through large arrays of low-cost passive elements and can be deployed across terrestrial and aerial platforms to improve coverage and link quality \cite{dao2021survey}.
However, conventional RISs are limited to reflection, restricting full-space coverage.
Simultaneously transmitting and reflecting RISs (STAR-RISs) overcome this by enabling full-space signal manipulation with enhanced coverage flexibility and spatial degrees of freedom \cite{ahmed2023survey}.
The resulting cascaded links, however, suffer from a double-fading effect that degrades reliability, particularly for IoT users far from the source.
Active STAR-RISs, which amplify the incident signal at each element, mitigate this double-fading effect and provide a cost- and energy-efficient alternative to conventional cooperative relays, thereby circumventing the noise/interference amplification, processing delay, and self-interference-cancellation burdens of amplify-and-forward, decode-and-forward, and full-duplex relays, respectively \cite{katiyar2011cooperative}.

In parallel, rate-splitting multiple access (RSMA) has emerged as a general framework that bridges and outperforms conventional orthogonal and non-orthogonal multiple access \cite{Clerckx0523aprimer}.
By splitting each user's message into a common part decoded by all users and a private part decoded individually \cite{mao2022rate}, and applying successive interference cancellation (SIC) at the receiver, RSMA flexibly manages multi-user interference, supports scheduling-free operation, and enables controlled information sharing that can enhance the confidentiality of private messages \cite{zhang2024enhancing}.
These properties make active STAR-RIS-assisted RSMA a compelling candidate for robust, energy-efficient IoT connectivity, while the double-fading cascaded channels and the near--far user imbalance introduce reliability and fairness challenges that motivate this work.

\vspace{-0.25cm}
\subsection{Literature Overview}
\label{Sect:Overview}

\begin{table*}[!t]
\vspace{-0.5cm}
\centering
\caption{Comparison of the proposed work with representative state-of-the-art STAR-RIS-assisted RSMA studies.}
\label{Tab:Comparison}
\footnotesize
\renewcommand{\arraystretch}{1.2}
\setlength{\tabcolsep}{4pt}
\begin{tabular}{|c|c|c|c|c|c|c|c|p{6.5cm}|}
\hline
\multirow{2}{*}{\textbf{Ref.}} & \multirow{2}{*}{\textbf{Year}} & \textbf{Active} & \textbf{Closed-form} & \textbf{Asymptotic} & \textbf{Model-} & \textbf{Fairness-} & \textbf{Data-} & \multirow{2}{*}{\textbf{Main objective / setup}} \\
 & & \textbf{STAR-RIS} & \textbf{OP/EC} & \textbf{DO/EE} & \textbf{based RA} & \textbf{oriented} & \textbf{driven RA} & \\
\hline\hline
\cite{karim2023performance}      & 2023 &            & \checkmark &            &            &            &                        & OP (energy-splitting/mode-switching)          \\ \hline
\cite{dhok2022rate}              & 2022 &            & \checkmark & \checkmark   &            &            &                        & OP/EC (correlated Rician) \\ \hline
\cite{gomes2025performance}      & 2025 &            & \checkmark &  \checkmark          &            &            &                        & OP (correlated, imperfect SIC)  \\ \hline
\cite{mondal2025performance}     & 2025 &            & \checkmark & \checkmark           &            &            &                        & OP/EC (sensing)      \\ \hline
\cite{tu2025active} & 2025 & \checkmark & \checkmark  & \checkmark  & & & & OP/EC (active)
\\ \hline
\cite{singh2025robust}           & 2025 & \checkmark & \checkmark & \checkmark           &            &            &                        & OP/throughput (UAV, active), learning-based performance predictions     \\ \hline
\cite{soleymani2025star}         & 2025 &            &            &            & \checkmark &            &                        & EE (finite-blocklength MU-MIMO)           \\ \hline
\cite{asif2025robust}            & 2025 &            &            &            & \checkmark &            &                        & Sum-rate (hardware impairments)     \\ \hline
\cite{katwe2023improved}         & 2023 &            &            &            & \checkmark &            &                        & Sum-throughput (uplink)         \\ \hline
\cite{liu2025star}               & 2025 &            &            &            & \checkmark &            &                        & Sensing SINR (sensing)            \\ \hline
\cite{maraqa2023optical}         & 2023 &            &            &            & \checkmark &            &                        & Sum-rate (optical/visible light)          \\ \hline
\cite{maghrebi2024deep}          & 2024 & \checkmark &            &            &            &            & \checkmark             & Sum-rate (DRL, cooperative)     \\ \hline
\cite{amiri2025resource}         & 2025 &            &            &            &            &            & \checkmark             & Sum-rate/harvested energy (meta-RL) \\ \hline\hline
\textbf{Ours} & \textbf{Now} & \checkmark & \checkmark & \checkmark & \checkmark & \checkmark & \checkmark &  \textbf{OP/EC fairness (active), low-overhead learning RA} \\ \hline
\end{tabular}
\end{table*}

We review the state of the art along three lines and summarize the comparison in Table~\ref{Tab:Comparison}. Throughout, ``low-overhead learning'' refers to supervised, training-light models that avoid the thousands of training episodes required by deep reinforcement learning (DRL).

\textit{Theoretical analysis}:
Existing information-theoretic studies of STAR-RIS-assisted RSMA IoT systems are largely confined to the passive mode.
Karim \etal \cite{karim2023performance} derived closed-form outage probability (OP) expressions under energy-splitting (ES) and mode-switching protocols and showed that RSMA can outperform non-orthogonal multiple access (NOMA) while quantifying the effects of transmit power, STAR-RIS size, and imperfect channel state information (CSI).
Dhok \etal \cite{dhok2022rate} analyzed spatially correlated Rician channels and provided OP and ergodic capacity (EC) in the infinite-blocklength regime together with block-error-rate and goodput results in the finite-blocklength regime.
Gomes \etal \cite{gomes2025performance} characterized the OP under correlated fading and imperfect SIC, revealing notable degradation under realistic SIC errors, while Mondal \etal \cite{mondal2025performance} extended the analysis to integrated sensing and communication (ISAC) with multiple targets, deriving OP/EC and highlighting sensing-communication tradeoffs.
For the active mode within this category, closed-form analyses remain limited: the work in \cite{tu2025active} derives OP/EC together with asymptotic analysis for an active STAR-RIS RSMA system, but it is confined to performance analysis without fairness-oriented model-based or data-driven RA, whereas \cite{singh2025robust} considers an unmanned-aerial-vehicle (UAV)-integrated network and provides OP and throughput under interference and hardware impairments, together with learning-based performance prediction, yet without a complete OP/EC, RA, and EE characterization.

\textit{Model-based RA}:
A second line formulates joint resource allocation (RA) to optimize rate-, reliability-, energy-, or sensing-related objectives.
Soleymani \etal \cite{soleymani2025star} employed block coordinate descent (BCD) to jointly design precoders, STAR-RIS configurations, and rate-splitting (RS) coefficients for EE maximization in finite-blocklength multi-user multiple-input multiple-output (MU-MIMO) systems.
Asif \etal \cite{asif2025robust} jointly optimized active precoders, RS coefficients, and passive beamforming for sum-rate maximization under hardware impairments and bounded CSI via successive convex approximation (SCA) and semidefinite relaxation (SDR).
For uplink, Katwe \etal \cite{katwe2023improved} maximized sum-throughput by jointly optimizing power allocation (PA), decoding order, and beamforming using BCD, SCA, and fractional programming, whereas Liu \etal \cite{liu2025star} maximized the sensing signal-to-interference-plus-noise ratio in an ISAC network via Dinkelbach's transform and SDR/majorization-minimization.
Beyond radio frequency, Maraqa \etal \cite{maraqa2023optical} solved sum-rate maximization for an optical STAR-RIS visible-light system using a sine-cosine metaheuristic.
These designs are predominantly sum-rate/energy-centric and do not address user fairness.

\textit{Data-driven RA}:
Learning-based RA remains nascent and relies on DRL.
Maghrebi \etal \cite{maghrebi2024deep} proposed an actor-critic DRL design for a cooperative active STAR-RIS RSMA system in which strong users relay the common stream to weak users, jointly optimizing precoders, STAR-RIS configurations, RS coefficients, and relay power for sum-rate maximization.
Amiri \etal \cite{amiri2025resource} adopted a meta-DRL scheme for a STAR-RIS simultaneous wireless information and power transfer RSMA system to balance sum rate and harvested energy.
Both target sum-rate/energy objectives, and DRL typically requires thousands of training episodes and is sensitive to environment changes, which limits real-time deployment.

\subsection{Novelty and Contributions}
\label{Sect:Novelty}

As summarized in Table~\ref{Tab:Comparison}, existing studies on STAR-RIS-assisted RSMA IoT systems remain incomplete in three respects.
First, information-theoretic frameworks are largely confined to passive STAR-RIS; for the active mode, closed-form analysis is available in \cite{tu2025active} and \cite{singh2025robust}, whereas \cite{maghrebi2024deep} employs it only for DRL-based RA, so no existing active-mode study couples a complete OP/EC, asymptotic, diversity-order/array-gain, and EE characterization with fairness-oriented model-based and low-overhead data-driven RA.
Second, existing model-based RA designs primarily adopt sum-rate/energy-centric objectives and rely on iterative optimization, whereas fairness-oriented formulations are largely overlooked and challenging to run in real time.
Third, learning-based RA mainly employs DRL, which incurs high training overhead and operational instability.
Motivated by these gaps, this work develops a unified theoretical-optimization-learning framework that (\textit{i}) provides comprehensive information-theoretic analysis, (\textit{ii}) enables fairness-enhancing optimization-based RA, and (\textit{iii}) introduces low-complexity data-driven RA models that approximate optimization solutions with improved scalability and real-time applicability.
We emphasize that the novelty lies not in treating analysis, optimization, or learning in isolation, each of which has appeared separately, but in their \emph{unified integration} for the active STAR-RIS RSMA IoT setting under a fairness objective, which has not been jointly addressed before, as contrasted along the active-mode, analysis-completeness, fairness, and learning-overhead axes in Table~\ref{Tab:Comparison}. The overall interaction among these three frameworks is summarized in the unified workflow of Fig.~\ref{Fig:ML_App}.
The main contributions of this work are summarized as follows:
\begin{itemize}
    \item \textbf{Comprehensive information-theoretic analysis}: We establish, for the first time, a \emph{complete and unified} analytical framework for active STAR-RIS-assisted RSMA IoT systems by deriving closed-form expressions for the OP and EC for scenarios both with and without direct links. We further perform asymptotic OP/EC analysis and characterize the diversity order (DO) and array gain, providing fundamental insights into system behavior across operating regimes. Moreover, to capture the tradeoff between transmission performance and power consumption, we investigate both throughput-based and spectral-based EE metrics, enabling a more comprehensive evaluation of the proposed design under practical energy expenditure.
    \item \textbf{Fairness-oriented model-based RA}: To enhance user fairness, we formulate OP/EC fairness-oriented RA frameworks that jointly optimize PA and RS coefficients under practical system constraints. The resulting non-convex problem is efficiently addressed via BCD and SCA, enabling user fairness-aware operation.
    \item \textbf{Scalable and real-time data-driven RA}: To overcome the real-time and scalability limitations of the model-based RA approaches, we propose a data-driven RA framework based on machine learning (ML) models, including a deep multi-output neural network (DMNN), a convolutional multi-output neural network (CMNN), and multi-output extreme gradient boosting (MXGB), to approximate optimization-based solutions. Numerical results demonstrate that the proposed MXGB model achieves the shortest execution time while maintaining performance comparable to ground-truth-based benchmarks, making it a compelling candidate for practical implementations.
    \item \textbf{Extensive simulations and validation}: We provide extensive Monte Carlo simulations to validate the accuracy of the derived analytical expressions and to verify the effectiveness of the proposed fairness-oriented and data-driven RA frameworks under various system settings.
\end{itemize}

\subsection{Organization and Notations}

The remainder of this paper is organized as follows.
Section~\ref{Sect:SystemModel} presents the network description and signal processing of the considered active STAR-RIS-assisted RSMA IoT systems.
Section~\ref{Sect:Analytical} develops the proposed information-theoretic analysis framework, including the derivations of OP and EC as well as asymptotic performance characterization.
Section~\ref{Sect:Model_based_RA} introduces the model-based RA framework for fairness-oriented optimization under practical constraints.
Section~\ref{Sect:Data-Driven} proposes the data-driven RA framework using ML models to predict RA configurations.
Section~\ref{Sect:Results} provides numerical results and discussions to validate the analytical expressions and evaluate the effectiveness of the proposed optimization and learning-based designs.
Finally, Section~\ref{Sect:Conclusion} presents the concluding remarks.

\textit{Mathematical Notations}: $(\cdot)^T$, $|\cdot|$, and $\mathbb{C}^{M \times 1}$ denote the transpose, the magnitude, and the $M$-dimensional complex space, respectively;
$\mathbb{E}(\cdot)$, ${\rm{Var}}(\cdot)$, and ${\rm{Pr}}(\cdot)$ are the expectation, variance, and probability operators;
and $F_X(\cdot)$ and $f_X(\cdot)$ are the cumulative distribution function (CDF) and probability density function (PDF) of a random variable $X$.
${\cal N}(\mu,\sigma^2)$ and ${\rm{Gamma}}(\xi,\Omega)$ denote the Gaussian (mean $\mu$, variance $\sigma^2$) and Gamma (shape $\xi$, scale $\Omega$) distributions, and ${\cal GG}(\alpha,d,p)$ the generalized Gamma distribution with shape parameters $\alpha,d$ and scale parameter $p$.
Moreover, ${\rm{G}}_{p,q}^{m,n}[\cdot]$, $\Upsilon(\cdot,\cdot)$, and $\Gamma(\cdot,\cdot)$ denote Meijer's G-function \cite[9.301]{gradshteyn2014table} and the lower and upper incomplete Gamma functions \cite[Eqs.~(8.350.1) and (8.350.2)]{gradshteyn2014table}, respectively.
Unless stated otherwise, $i = \overline{1,N}$, $\ell \in \{ {\rm{n}},{\rm{f}} \}$, and $\ell' \in \{ {\rm{n}},{\rm{f}} \}, \ell'\ne \ell$ are assumed throughout.

\section{Network Description}
\label{Sect:SystemModel}

\begin{figure}[!t]
\centering
\includegraphics[width=.9\linewidth]{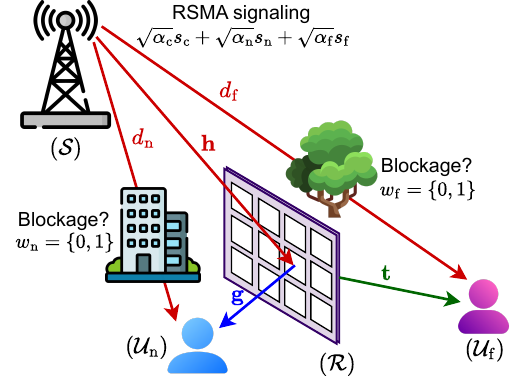}
\caption{Downlink STAR-RIS-assisted RSMA IoT systems.}
\vspace{-0.25cm}
\label{Fig:SystemModel}
\end{figure}

This work examines a downlink RSMA IoT system supported by a STAR-RIS, as illustrated in Fig.~\ref{Fig:SystemModel}.\footnote{The considered setup is tailored, but not limited, to IoT communications in several respects: (\textit{i})~the low-cost, energy-constrained, and coverage-limited nature of IoT links motivates the active STAR-RIS, which provides low-power full-space coverage and amplifies the double-fading cascaded channel that severely affects far IoT devices; (\textit{ii})~the strong near--far disparity among heterogeneous IoT devices makes user fairness, rather than sum rate, the operationally relevant objective; (\textit{iii})~IoT gateways demand real-time, low-complexity resource allocation on constrained hardware, motivating the lightweight learning-based inference over iterative optimization; and (\textit{iv})~the scheduling-free operation and controlled common-stream decoding of RSMA match the sporadic, low-latency access of IoT and offer the potential to scale toward denser device populations through user grouping/pairing (Remark~\ref{Remark:MultiUser}).}
An information source $(\cal S)$ communicates with two IoT users, namely the near user ${\cal U}_{\rm{n}}$ and the far user ${\cal U}_{\rm{f}}$, through a STAR-RIS $(\cal R)$
consisting of $N$ elements.
Throughout the analysis, we assume that ${\cal U}_{\rm{n}}$ lies within the reflection region of the STAR-RIS, whereas ${\cal U}_{\rm{f}}$ is positioned in its transmission region.
The STAR-RIS operates according to an ES protocol, which allows each element to flexibly divide the incoming energy into refracted and reflected components.
Compared to mode-switching or time-splitting protocols, the ES approach provides greater tuning capability of the corresponding coefficients \cite{mu2021simultaneously}.

Generally, both direct and cascaded links between $\cal S$ and  ${\cal U}_\ell$, $\forall \ell \in \{ {\rm{n}},{\rm{f}} \}$, are considered.
When the direct paths are blocked or severely shadowed, the two RSMA users rely primarily on the STAR-RIS, which then serves as an alternative communication link. In contrast, when the direct links are available, the system leverages cooperative transmission that combines both the cascaded and direct channels.
To capture the presence or absence of a direct link, we introduce the binary parameter $w_\ell \in \{ 0,1\}$, where $w_\ell = 1$ indicates the existence of a direct link and $w_\ell = 0$ indicates that only the cascaded (STAR-RIS-assisted) path is available.

Let ${\bf{h}} = [h_1,...,h_N]^T \in \mathbb{C}^{N\times1}$,
${\bf{g}} = [g_1,...,g_N]^T \in \mathbb{C}^{N\times1}$,
${\bf{t}} = [t_1,...,t_N]^T \in \mathbb{C}^{N\times1}$,
and $d_\ell \in \mathbb{C}^{1\times 1}$
denote the channel coefficients of the links
$\cal S \to \cal R$,
$\cal R \to \cal U_{\rm{n}}$,
$\cal R \to \cal U_{\rm{f}}$,
and $\cal S \to \cal U_\ell$, respectively.
All individual coefficients follow independent and identically distributed (i.i.d.) Rayleigh fading channels, such that
$|{{h}}_i|^2$, $|{{g}}_i|^2$, $|{{t}}_i|^2$, and $|{d_\ell}|^2$
exhibit exponential distributions with mean values characterized by
$\lambda_h = \mathbb{E}\{ |{{h}}_i|^2 \} = {\rm{d}}_h^{-\eta}$,
$\lambda_g = \mathbb{E} \{ |{{g}}_i|^2 \} = {\rm{d}}_g^{-\eta}$,
$\lambda_t = \mathbb{E} \{ |{{t}}_i|^2 \} = {\rm{d}}_t^{-\eta}$,
and $\lambda_{\ell} = \mathbb{E} \{ |{d_\ell}|^2 \} = {\rm{d}}_{\ell}^{-\eta}$,
where ${\rm{d}}_{\cal X}$, $\forall {\cal X} \in \{ h,g,t,\ell \}$, represents the physical distance of the corresponding channel
and $\eta$ denotes the path-loss exponent \cite{tu2022performance,tu2023short2}.
Perfect CSI is assumed to be available at all network nodes.

During downlink transmission, $\cal S$ employs RSMA by generating a superimposed signal
$\sqrt {{\alpha _{\rm{c}}}} {s_{\rm{c}}} + \sqrt {{\alpha _{\rm{n}}}} {s_{\rm{n}}} + \sqrt {{\alpha _{\rm{f}}}} {s_{\rm{f}}}$,
where $s_{\rm{c}}$ is the common stream intended for both users,
and $s_{\rm{n}}$ and $s_{\rm{f}}$ are the private streams for $\cal U_{\rm{n}}$ and $\cal U_{\rm{f}}$, respectively.
The PA coefficients $\alpha_{\rm{c}}$ and $\alpha_\ell$ are chosen such that ${\alpha _{\rm{c}}} + {\alpha _{\rm{n}}} + {\alpha _{\rm{f}}} = 1$,
while successful recovery of ${s_{\rm{c}}}$ at $\cal U_\ell$ before retrieving its own private stream $s_\ell$ imposes the condition ${\alpha _{\rm{c}}} \ge \alpha_\ell$, $\forall \ell$ \cite{vu2025study}.
The received signals at the two users can be written as
\setlength{\belowdisplayskip}{3pt} \setlength{\belowdisplayshortskip}{3pt}
\setlength{\abovedisplayskip}{3pt} \setlength{\abovedisplayshortskip}{3pt}
\begin{align*}
    {y_{\rm{n}}} &= \left( {{w_{\rm{n}}}{d_{\rm{n}}} + {{\bf{g}}^T}{{\bf{\Phi }}_{\rm{R}}}{\bf{h}}} \right)\left( {\sqrt {{\alpha _{\rm{c}}}} {s_{\rm{c}}} + \sqrt {{\alpha _{\rm{n}}}} {s_{\rm{n}}} + \sqrt {{\alpha _{\rm{f}}}} {s_{\rm{f}}}} \right) + {n_{\rm{n}}},\\
    {y_{\rm{f}}} &= \left( {{w_{\rm{f}}}{d_{\rm{f}}} + {{\bf{t}}^T}{{\bf{\Phi }}_{\rm{T}}}{\bf{h}}} \right)\left( {\sqrt {{\alpha _{\rm{c}}}} {s_{\rm{c}}} + \sqrt {{\alpha _{\rm{n}}}} {s_{\rm{n}}} + \sqrt {{\alpha _{\rm{f}}}} {s_{\rm{f}}}} \right) + {n_{\rm{f}}},
\end{align*}
where $\mathbb{E}\{ {{{\left| {{s_{\rm{c}}}} \right|}^2}} \} = \mathbb{E}\{ {{{\left| {{s_\ell}} \right|}^2}} \} = P_S$ is the average transmit signal power
and
${n_\ell } \sim {\cal C}{\cal N}\left( {0,{\sigma ^2}} \right)$ is the additive white Gaussian noise at $\cal U_\ell$.
The reflection and transmission matrices of the STAR-RIS under ES are written as
${{\bf{\Phi }}_{\rm{R}}} = {\rm{diag}}\big(\sqrt {\delta _1^{\rm{R}}\beta _1^{\rm{R}}} {e^{j{\varphi _1}}},...,\sqrt {\delta _N^{\rm{R}}\beta _N^{\rm{R}}} {e^{j{\varphi _N}}}\big)$
and
${{\bf{\Phi }}_{\rm{T}}} = {\rm{diag}}\big(\sqrt {\delta _1^{\rm{T}}\beta _1^{\rm{T}}} {e^{j{\nu _1}}},...,\sqrt {\delta _N^{\rm{T}}\beta _N^{\rm{T}}} {e^{j{\nu _N}}}\big)$,
where the ES coefficients $\beta _i^{\rm{R}}$ and $\beta _i^{\rm{T}}$ satisfy $\beta _i^{\rm{R}} + \beta _i^{\rm{T}} = 1$,
$\{ {\varphi _i},{\nu _i}\}$ denotes the phase-shift coefficient,
and $\{ \delta _i^{\rm{R}},\delta _i^{\rm{T}}\} \ge 1$ is the amplification gain \cite{vu2024hybrid,vu2025aerial}.
Notably, $\cal R$ operates in active mode when $\{ \delta _i^{\rm{R}},\delta _i^{\rm{T}}\} > 1$, and in passive mode when $\delta _i^{\rm{R}} = \delta _i^{\rm{T}} = 1$.

Upon receiving the signal $y_\ell$, each user performs SIC following the descending order of the PA coefficients to recover both the common and private streams \cite{vu2025study}.
Specifically, each user first decodes the common stream ${s_{\rm{c}}}$ and then proceeds to decode its own private stream $s_\ell$, while
treating the other user’s private stream $s_{\ell'}$, $\forall \ell' \in \{ {\rm{n}},{\rm{f}} \}$, $\ell' \ne \ell$, as interference.
Assuming perfect CSI, the STAR-RIS phase-shift parameters $\{ {\varphi _i},{\nu _i}\}$ are configured to maximize the received signal strength.
This is achieved by setting
${\varphi _i} = \arg \left( {{w_{\rm{n}}}{d_{\rm{n}}}} \right) - \arg \left( {{g_i}} \right) - \arg \left( {{h_i}} \right)$
and
${\nu _i} = \arg \left( {{w_{\rm{f}}}{d_{\rm{f}}}} \right) - \arg \left( {{t_i}} \right) - \arg \left( {{h_i}} \right)$
\cite{vu2024hybrid,vu2025aerial}.
Consequently, the achievable rates for decoding the common stream ${s_{\rm{c}}}$ and the private stream $s_\ell$ at user $\cal U_\ell$ can be expressed, respectively, as
\begin{align}\label{eq:R_cl}
    {\Re _{{\rm{c}},\ell }} &= {\tau _\ell }{\log _2}\left( {1 + \frac{{{\kappa _\ell }{\alpha _{\rm{c}}}\bar \gamma }}{{{\kappa _\ell }(1 - {\alpha _{\rm{c}}})\bar \gamma  + 1}}} \right), \forall \ell,\\
    \label{eq:R_l}
    {\Re _\ell } &= {\log _2}\left( {1 + \frac{{{\kappa _\ell }{\alpha _\ell }\bar \gamma }}{{{\kappa _\ell }{\alpha _{\ell '}}\bar \gamma  + 1}}} \right), \forall \ell, \ell',
\end{align}
where
${\tau _\ell }$ denotes the RS coefficient of $\cal U_\ell$ from the common stream ${s_{\rm{c}}}$, satisfying ${\tau _{\rm{n}}} + {\tau _{\rm{f}}} = 1$,
$\bar \gamma  = {P_S}/{\sigma ^2}$ is the average transmit signal-to-noise ratio (SNR),
and
\begin{align}
    {\kappa _{\ell}} =
    \begin{cases}
        {\Big| {{w_{\rm{n}}}\left|{d_{\rm{n}}}\right| + \sum\nolimits_{i = 1}^N {\sqrt {\delta _i^{\rm{R}}\beta _i^{\rm{R}}} \left| {{g_i}} \right|\left| {{h_i}} \right|} } \Big|^2}, &  {\rm{if}}\,\, \ell = {\rm{n}},\\
        {\Big| {{w_{\rm{f}}} \left|{d_{\rm{f}}}\right| + \sum\nolimits_{i = 1}^N {\sqrt {\delta _i^{\rm{T}}\beta _i^{\rm{T}}} \left| {{t_i}} \right|\left| {{h_i}} \right|} } \Big|^2}, & {\rm{if}}\,\, \ell = {\rm{f}}.
    \end{cases}
\end{align}

\begin{remark}\label{Remark:MultiUser}
The two-user near/far configuration is the canonical RSMA setting that isolates the near--far fairness problem and admits tractable closed-form OP/EC analysis. For massive-connectivity IoT, the framework extends via user grouping/pairing: devices are clustered (e.g., by large-scale gain) into near/far groups served on orthogonal resources, and per-group RSMA with hierarchical SIC is applied. Each group then reduces to the two-user problem, so the derived expressions and the algorithms of Sections~\ref{Sect:Analytical} and \ref{Sect:Model_based_RA} are reused per group. The main obstacle is the intra-group SIC ordering, which is combinatorial (up to $K!$ orders for $K$ users) but is well handled by the standard heuristic of ordering by descending PA coefficients, i.e., an ${\cal O}(K\log K)$ sort that is near-optimal for near/far-structured channels. A rigorous massive-connectivity treatment additionally requires joint user clustering, SIC-ordering optimization, and a re-derivation of the multi-user expressions, and is left as future work (Section~\ref{Sect:FutureWork}).
\end{remark}

\begin{assumption}\label{Assumption:1}
    It is worth mentioning that when $\{ \delta _i^{\rm{R}},\delta _i^{\rm{T}}\} > 1$, \eqref{eq:R_cl} and \eqref{eq:R_l} are derived under the assumption that the self-interference (SI) of the active STAR-RIS is neglected.
    The SI power generated by active STAR-RIS elements, which is typically around $-80$ dBm \cite{vu2025aerial}, is insignificant when compared to the inter-message interference resulting from the source transmission power $P_S$, generally on the order of $30$ dBm or higher. This represents roughly a $10^{11}$-fold difference in power. Consequently, in RSMA communication systems, inter-message interference overwhelmingly dominates, allowing the SI introduced by the STAR-RIS to be safely neglected.
    This assumption is stated within the practical operating region $\delta _i^{\rm{R}} = \delta _i^{\rm{T}} = \delta \le \delta_{\rm sat}$, where $\delta_{\rm sat}$ denotes the amplifier saturation gain, so that the active elements operate in their linear regime. A quantitative rate-loss bound substantiating this assumption is derived in Section~\ref{Sect:5.1}, while the complementary high-gain regime $\delta > \delta_{\rm sat}$, in which residual-SI accumulation and amplifier nonlinearity become the dominant impairments, is examined separately in Section~\ref{Sect:Residual_SI}.
\end{assumption}

\section{Information-Theoretic Analysis Framework}
\label{Sect:Analytical}
In this section, we derive closed-form expressions for the OP and EC performance of RSMA users, from which the asymptotic behavior, DO, array gain, and throughput-based and spectral-based EE are analyzed.

\subsection{Outage Probability Framework}\label{Sect:OP}

\subsubsection{Outage Probability Analysis}
The OP measures the likelihood that a user fails to decode either the common stream or its private stream \cite{vu2025study}.
For user $\cal U_\ell$, the OP is mathematically expressed as
\begin{align}\label{eq:OP_l_original}
    {\rm{O}}{{\rm{P}}_\ell } &= 1 - \Pr \left( {{\Re _{{\rm{c}},\ell }} \ge R_{\rm{c}}^{{\rm{th}}},{\Re _\ell } \ge R_\ell ^{{\rm{th}}}} \right)\nonumber\\
    &= 1 - \Pr \left\{ {{\kappa _\ell } \ge \max \left( {{\wp _{{\rm{c}},\ell }},{\wp _\ell }} \right)} \right\}\nonumber\\
    & = {F_{{\kappa _\ell }}}\left( {\max \left( {{\wp _{{\rm{c}},\ell }},{\wp _\ell }} \right)} \right),
\end{align}
where $R_{\rm{c}}^{{\rm{th}}}$ and $R_\ell ^{{\rm{th}}}$ are the target rates for ${s_{\rm{c}}}$ and ${s_{\ell}}$, respectively,
${\wp _{{\rm{c}},\ell }} = \gamma _{{\rm{c}},\ell }^{{\rm{th}}}{[ {{\alpha _{\rm{c}}}\bar \gamma  - \gamma _{{\rm{c}},\ell }^{{\rm{th}}}\left( {1 - {\alpha _{\rm{c}}}} \right)\bar \gamma } ]^{ - 1}}$,
${\wp _\ell } = \gamma _\ell ^{{\rm{th}}}{[ {{\alpha _\ell }\bar \gamma  - \gamma _\ell ^{{\rm{th}}}{{\alpha }_{\ell'} }\bar \gamma } ]^{ - 1}}$,
$\gamma _{{\rm{c}},\ell }^{{\rm{th}}} = {2^{R_{\rm{c}}^{{\rm{th}}}/{\tau _\ell }}} - 1$,
and
$\gamma _\ell ^{{\rm{th}}} = {2^{R_\ell ^{{\rm{th}}}}} - 1$.

To derive \eqref{eq:OP_l_original}, we determine the statistical characterization of ${{\kappa _\ell }}$ through the following proposition.

\begin{proposition}[Statistical Characterization of ${{\kappa _\ell }}$]
\label{proposition:1}
    The random variable ${{\kappa _\ell }}$ follows an approximate generalized Gamma distribution with first shape parameter $\xi_\ell$, second shape parameter $1/2$, and scale parameter $\Omega_\ell^2$, i.e., ${\kappa _\ell } \sim {\cal G}{\cal G}\left( {{\xi _\ell },1/2,\Omega _\ell ^2} \right)$.
    Accordingly, its CDF and PDF are given, respectively, by \cite{do2021multi}
    \begin{align}\label{eq:CDF_kappa}
        {F_{{\kappa _\ell }}}\left( {x;{\xi _\ell },\frac{1}{2},\Omega _\ell ^2} \right) &= \frac{1}{{\Gamma \left( {{\xi _\ell }} \right)}}\Upsilon \left( {{\xi _\ell },\frac{{\sqrt x }}{{\Omega _\ell ^{}}}} \right),x \ge 0,
        \\ \label{eq:PDF_kappa}
        {f_{{\kappa _\ell }}}\left( {x;{\xi _\ell },\frac{1}{2},\Omega _\ell ^2} \right) &= \frac{{{x^{{\xi _\ell }/2 - 1}}}}{{2\Gamma \left( {{\xi _\ell }} \right)\Omega _\ell ^{{\xi _\ell }}}}\exp \left( { - \frac{{\sqrt x }}{{\Omega _\ell ^{}}}} \right),x \ge 0,
    \end{align}
    where
    ${\xi _\ell } = \mu _\ell ^2/{\chi _\ell }$,
    $\Omega _\ell ^{} = {\chi _\ell }/{\mu _\ell }$,
    ${\mu _{\rm{n}}} = \frac{{{w_{\rm{n}}}}}{2}\sqrt {\pi {\lambda _{\rm{n}}}}  + \frac{\pi }{4}\sqrt {{\lambda _g}{\lambda _h}} \beta _{{\Sigma _1}}^{\rm{R}}$, ${\mu _{\rm{f}}} = \frac{{{w_{\rm{f}}}}}{2}\sqrt {\pi {\lambda _{\rm{f}}}}  + \frac{\pi }{4}\sqrt {{\lambda _t}{\lambda _h}} \beta _{{\Sigma _1}}^{\rm{T}}$,
    ${\chi _{\rm{n}}} = \frac{{4 - \pi }}{4}w_{\rm{n}}^2{\lambda _{\rm{n}}} + \frac{{16 - {\pi ^2}}}{{16}}{\lambda _g}{\lambda _h}\beta _{{\Sigma _2}}^{\rm{R}}$,
    ${\chi _{\rm{f}}} = \frac{{4 - \pi }}{4}w_{\rm{f}}^2{\lambda _{\rm{f}}} + \frac{{16 - {\pi ^2}}}{{16}}{\lambda _t}{\lambda _h}\beta _{{\Sigma _2}}^{\rm{T}}$,
    $\beta _{{\Sigma _1}}^{\rm{R}} = \sum\nolimits_{i = 1}^N {\sqrt {\delta _i^{\rm{R}}\beta _i^{\rm{R}}} } $,
    $\beta _{{\Sigma _1}}^{\rm{T}} = \sum\nolimits_{i = 1}^N {\sqrt {\delta _i^{\rm{T}}\beta _i^{\rm{T}}} } $,
    $\beta _{{\Sigma _2}}^{\rm{R}} = \sum\nolimits_{i = 1}^N {\delta _i^{\rm{R}}\beta _i^{\rm{R}}} $,
    and $\beta _{{\Sigma _2}}^{\rm{T}} = \sum\nolimits_{i = 1}^N {\delta _i^{\rm{T}}\beta _i^{\rm{T}}} $.
    The analytical correctness of this proposition is validated through Monte Carlo simulations in Fig.~\ref{Fig:Graphical_Validation}.
\begin{proof}
Let ${\kappa _{\ell,1}} = \left| {{d_{\ell}}} \right|$,
${\kappa _{{\rm{n}},2}} = \sum\nolimits_{i = 1}^N {\sqrt {\delta _i^{\rm{R}}\beta _i^{\rm{R}}} \left| {{g_i}} \right|\left| {{h_i}} \right|}$,
${\kappa _{{\rm{f}},2}} = \sum\nolimits_{i = 1}^N {\sqrt {\delta _i^{\rm{T}}\beta _i^{\rm{T}}} \left| {{t_i}} \right|\left| {{h_i}} \right|}$,
and
${\kappa _{{\ell},3}} = {w_{\ell}}{\kappa _{{\ell},1}} + {\kappa _{{\ell},2}}$.
Based on the moment-matching method \cite{tu2024multihop,do2021multi}, ${\kappa _{{\ell},3}}$ can be approximated as a Gamma distribution with shape and scale parameters as
${\xi _\ell } = {{\mathbb{E}{{\{ {\kappa _{\ell ,3}}\} }^2}}}\big/{{{\rm{Var}}\{ {\kappa _{\ell ,3}}\} }} \buildrel \Delta \over = {{\mu _\ell ^2}}/{{{\chi _\ell }}}$
and
$\Omega _\ell ^{} = {{{\rm{Var}}\{ {\kappa _{\ell ,3}}\} }}\big/{{\mathbb{E}\{ {\kappa _{\ell ,3}}\} }} \buildrel \Delta \over = {{{\chi _\ell }}}/{{{\mu _\ell }}}$, respectively,
where
\begin{align*}
    \hspace{-0.15cm}{\mu _{\ell}} &= \mathbb{E}\left\{ {{\kappa _{{\ell},3}}} \right\} = {w_{\ell}}\mathbb{E}\left\{ {{\kappa _{{\ell},1}}} \right\} + \mathbb{E}\left\{ {{\kappa _{{\ell},2}}} \right\} \nonumber\\
    &= \frac{{{w_{\ell}}}}{2}\sqrt {\pi {\lambda _{\ell}}} +
    \begin{cases}
        \frac{\pi }{4}\sqrt {{\lambda _g}{\lambda _h}} \sum\nolimits_{i = 1}^N {\sqrt {\delta _i^{\rm{R}}\beta _i^{\rm{R}}} } , & \hspace{-0.15cm} {\rm{if}}\, \ell = {\rm{n}},\\
        \frac{\pi }{4}\sqrt {{\lambda _t}{\lambda _h}} \sum\nolimits_{i = 1}^N {\sqrt {\delta _i^{\rm{T}}\beta _i^{\rm{T}}} } , & \hspace{-0.15cm} {\rm{if}}\, \ell = {\rm{f}},
    \end{cases}
    \\
    \hspace{-0.15cm}{\chi _{\ell}} &= {\rm{Var}}\left\{ {{\kappa _{{\ell},3}}} \right\} = w_{\ell}^2{\rm{Var}}\left\{ {{\kappa _{{\ell},1}}} \right\} + {\rm{Var}}\left\{ {{\kappa _{{\ell},2}}} \right\} \nonumber\\
    &= w_{\ell}^2\left( {\mathbb{E}\left\{ {\kappa _{{\ell},1}^2} \right\} - \mathbb{E}{{\left\{ {{\kappa _{{\ell},1}}} \right\}}^2}} \right) + \mathbb{E}\left\{ {\kappa _{{\ell},2}^2} \right\} - \mathbb{E}{\left\{ {{\kappa _{{\ell},2}}} \right\}^2} \nonumber\\
    &= \frac{{4 - \pi }}{4}w_{\ell}^2{\lambda _{\ell}} +
    \begin{cases}
        \frac{{16 - {\pi ^2}}}{{16}}{\lambda _g}{\lambda _h} \sum\nolimits_{i = 1}^N {\delta _i^{\rm{R}}\beta _i^{\rm{R}}}, & \hspace{-0.15cm} {\rm{if}}\, \ell = {\rm{n}},\\
        \frac{{16 - {\pi ^2}}}{{16}}{\lambda _t}{\lambda _h} \sum\nolimits_{i = 1}^N {\delta _i^{\rm{T}}\beta _i^{\rm{T}}}, & \hspace{-0.15cm} {\rm{if}}\, \ell = {\rm{f}}.
    \end{cases}
\end{align*}
When ${\kappa _{{\ell},3}} \sim {\rm{Gamma}} \left( {{\xi _{\ell}},\Omega _{\ell}^{}} \right)$,
we have
${\kappa _{\ell}} \buildrel \Delta \over = {\left| {{\kappa _{{\ell},3}}} \right|^2} \sim {\cal G}{\cal G}\left( {{\xi _{\ell} },1/2,\Omega _{\ell} ^2} \right)$, whose CDF and PDF are given as \eqref{eq:CDF_kappa} and \eqref{eq:PDF_kappa}, respectively.
The proof is concluded.
\end{proof}

\end{proposition}

\begin{figure}[!t]
\centering
\subfloat[CDF]{\includegraphics[width=0.49\linewidth]{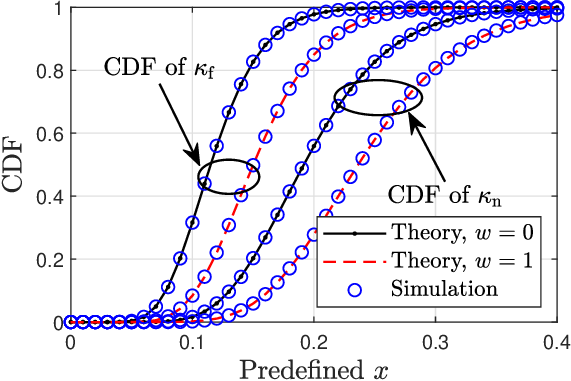}}\hfill
\subfloat[PDF]{\includegraphics[width=0.49\linewidth]{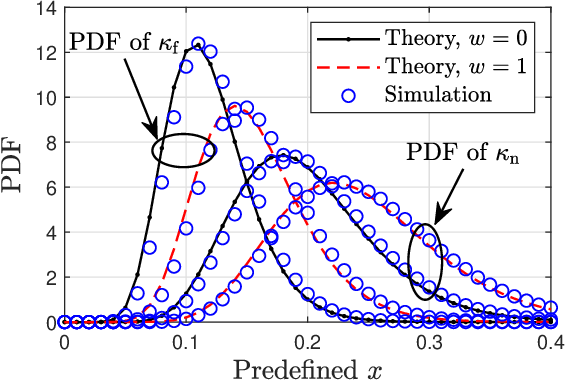}}
\caption{Graphical validation for the distributions of ${{\kappa _\ell }}$ in Proposition~\ref{proposition:1}, where we assume $N = 30$,
$\beta_i^{\rm{R}}=0.6$,
$\beta_i^{\rm{T}}=0.4$,
$\delta_i^{\rm{R}}= \delta_i^{\rm{T}}=10$,
$\lambda_{\rm{n}} = \lambda_{\rm{f}} = 0.003$,
$\lambda_h$ = 0.0035,
$\lambda_g = 0.025$,
and $\lambda_t = 0.01$.
The simulations are based on the Monte Carlo method with $10^6$ channel realizations.
}
\label{Fig:Graphical_Validation}
\end{figure}

\begin{figure}[!t]
\centering
\includegraphics[width=\linewidth]{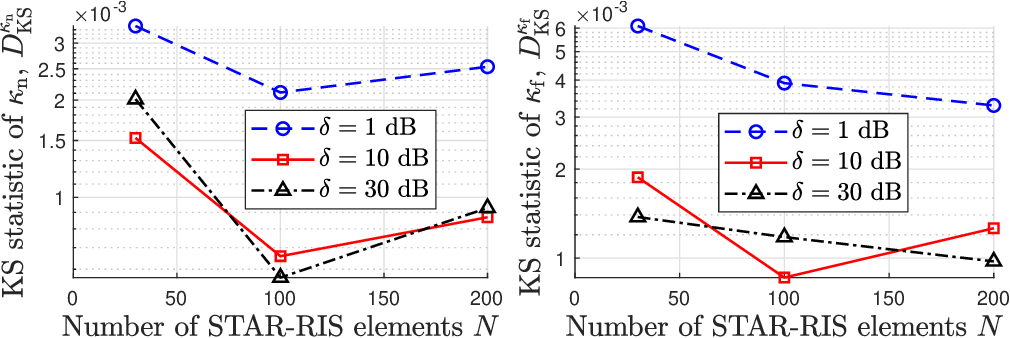}
\caption{KS statistic versus the number of STAR-RIS elements $N$ for different amplification gains $\delta$: (a)~$\kappa_{\rm n}$ and (b)~$\kappa_{\rm f}$.}
\label{Fig:KS_Statistic}
\end{figure}
To complement the graphical validation in Fig.~\ref{Fig:Graphical_Validation}, Fig.~\ref{Fig:KS_Statistic} quantitatively assesses the accuracy of the generalized-Gamma approximation in Proposition~\ref{proposition:1} using the Kolmogorov--Smirnov (KS) statistic, evaluated between the empirical CDF from $10^6$ Monte Carlo realizations and the analytical CDF in \eqref{eq:CDF_kappa}, defined as $D_{\mathrm{KS}}^{\kappa_\ell}=\sup_x|\hat{F}_{\kappa_\ell}(x)-F_{\kappa_\ell}(x)|$.
The KS statistic is examined across STAR-RIS sizes $N\in\{30,100,200\}$ and amplification gains $\delta_i^{\rm{R}} = \delta_i^{\rm{T}} = \delta\in\{1,10,30\}$~dB for both the near and far users. As observed in Fig.~\ref{Fig:KS_Statistic}, the KS statistic remains small over the entire range, on the order of $10^{-3}$, confirming that the approximation is tight across different STAR-RIS configurations and amplification gains.
It is also worth noting that, for a finite number of samples $N_{\rm{trial}}$, the KS statistic does not vanish but is lower bounded by an order of $1/\sqrt{N_{\rm{trial}}}$ even under an exact distributional match; the residual values in Fig.~\ref{Fig:KS_Statistic} are at this sampling-limited level, indicating that the deviation is dominated by finite-sample effects rather than by the approximation itself. These results validate the robustness of the generalized-Gamma approximation adopted throughout the analytical framework.

From Proposition~\ref{proposition:1}, ${\rm{O}}{{\rm{P}}_\ell }$ in \eqref{eq:OP_l_original} can be readily computed as
\begin{align}\label{eq:OP_l_final}
    {\rm{O}}{{\rm{P}}_\ell } = \frac{1}{{\Gamma \left( {{\xi _\ell }} \right)}}\Upsilon \left( {{\xi _\ell },\frac{1}{{{\Omega _\ell }}}\sqrt {\max \left( {{\wp _{{\rm{c}},\ell }},{\wp _\ell }} \right)} } \right),
\end{align}
where $\max \left( {{\wp _{{\rm{c}},\ell }},{\wp _\ell }} \right) \ge 0$; otherwise, ${\rm{O}}{{\rm{P}}_\ell } = 1$.
This yields the tight relationships such that
\begin{align}\label{eq:constraint1}
    {\alpha _{\rm{c}}} &> \gamma _{\rm{c},\ell }^{{\rm{th}}}\left( {1 - {\alpha _{\rm{c}}}} \right),\forall \ell ,\\
    \label{eq:constraint2}
    {\alpha _\ell } &> \gamma _\ell ^{{\rm{th}}}{\alpha _{\ell '}}, \forall \ell ,\ell '.
\end{align}

\subsubsection{Asymptotic Outage Probability}
In the high-SNR regime, i.e., $\bar \gamma \to \infty$, we have $\max \left( {{\wp _{{\rm{c}},\ell }},{\wp _\ell }} \right) \to 0$.
By applying the series expansion of the lower incomplete Gamma function \cite[Eq.~(8.354.1)]{gradshteyn2014table}, the asymptotic OP corresponding to \eqref{eq:OP_l_final} is calculated as
${\widetilde {{\rm{OP}}}_\ell } = \frac{1}{{\Gamma \left( {{\xi _\ell }} \right)}}\sum\limits_{k = 0}^\infty  {\frac{{{{\left( { - 1} \right)}^k}}}{{k!\left( {k + {\xi _\ell }} \right)}}} {\left( {\frac{1}{{{\Omega _\ell }}}\sqrt {\max \left( {{\wp _{{\rm{c}},\ell }},{\wp _\ell }} \right)} } \right)^{k + {\xi _\ell }}}$.
When $\max \left( {{\wp _{{\rm{c}},\ell }},{\wp _\ell }} \right) \to 0$,
this simplifies to
\begin{align}\label{eq:OP_l_asymptotic}
    {\widetilde {{\rm{OP}}}_\ell } = \frac{1}{{\Gamma \left( {{\xi _\ell } + 1} \right)}}{\left( {\frac{1}{{{\Omega _\ell }}}\sqrt {\max \left( {{\wp _{{\rm{c}},\ell }},{\wp _\ell }} \right)} } \right)^{{\xi _\ell }}}.
\end{align}

\begin{proposition}[Diversity Order and Array Gain]\label{proposition:diversity_order}
    To observe the DO\footnote{The DO represents the rate at which the OP decreases with increasing SNR. A larger DO corresponds to a steeper OP decay curve, which implies enhanced robustness against fading and better reliability in the high-SNR range.} of the proposed system model, \eqref{eq:OP_l_asymptotic} can be reformulated as
    \begin{align}\label{eq:DO}
        {\widetilde {{\rm{OP}}}_\ell } = {G_\ell }{{\bar \gamma }^{ - {\xi _\ell }/2}},
    \end{align}
    where ${G_\ell } = \frac{1}{{\Gamma \left( {{\xi _\ell } + 1} \right)\Omega _\ell ^{{\xi _\ell }}}}\max {\Big( {\frac{{\gamma _{{\rm{c}},\ell }^{{\rm{th}}}}}{{{\alpha _{\rm{c}}} - \gamma _{{\rm{c}},\ell }^{{\rm{th}}}\left( {1 - {\alpha _{\rm{c}}}} \right)}},\frac{{\gamma _\ell ^{{\rm{th}}}}}{{{\alpha _\ell } - \gamma _\ell ^{{\rm{th}}}{\alpha _{\ell '}}}}} \Big)^{{\xi _\ell }/2}}$.
    Accordingly, the system DO is obtained as
    ${D_\ell } = -\mathop {\lim }\limits_{\bar \gamma  \to \infty } \,\,\log \big( {{{\widetilde {{\rm{OP}}}}_\ell }} \big)/\log \left( {\bar \gamma } \right) =  - \mathop {\lim }\limits_{\bar \gamma  \to \infty } \,\,\Big\{ \frac{{\log ({G_\ell })}}{{\log \left( {\bar \gamma } \right)}} + \frac{{\log ({{\bar \gamma }^{ - {\xi _\ell }/2}})}}{{\log \left( {\bar \gamma } \right)}} \Big\} = {\xi _\ell }/2$.
    For the array gain\footnote{The array gain represents the horizontal displacement of the OP curve along the SNR axis, quantifying the additional SNR required to reach a given outage level at the same DO. A higher array gain indicates that the same level of reliability can be achieved with lower transmit power, reflecting improved fading mitigation efficiency.}, \eqref{eq:DO} can be rearranged as
    ${\widetilde {{\rm{OP}}}_\ell } = {G_\ell }{{\bar \gamma }^{ - D_\ell}} =  (A_\ell {\bar \gamma })^{ - D_\ell}$, where $A_\ell = G_\ell^{-1/D_\ell} = G_\ell^{-2/\xi_\ell}$ defines the system array gain.

\end{proposition}

\begin{remark}\label{remark:OP}
    The findings from Proposition~\ref{proposition:diversity_order} indicate that the diversity order is directly governed by $\xi_\ell$.
    Increasing the number of STAR-RIS elements or enhancing the amplification gain directly raises $\xi_\ell$, leading to a noticeably steeper OP curve.
    Consequently, performance can be improved by deploying more STAR-RIS elements or by increasing the amplification gain.
    Additionally, when $D_\ell \ne 0$, the OP continues to decline even in the high-SNR region.
    Meanwhile, the horizontal placement of the OP curve (i.e., the array gain) is influenced by several system parameters, including the number of STAR-RIS elements, the amplification gain, target data rates, and the PA coefficients.
\end{remark}

\subsection{Ergodic Capacity Framework}

\subsubsection{Ergodic Capacity Analysis}

The EC represents the long-term average of the achievable data rates that can be supported by the channel with a negligible OP \cite{tu2024irs}.
Mathematically, given ${\Re _{{\rm{c}},\ell }}$ and ${\Re _\ell }$ in \eqref{eq:R_cl} and \eqref{eq:R_l}, respectively, the EC of $\cal U_\ell$ is expressed as
\begin{align}\label{eq:EC_l}
    &{\rm{E}}{{\rm{C}}_\ell } = \mathbb{E}\left\{ {{\Re _{{\rm{c}},\ell }}} \right\} + \mathbb{E}\left\{ {{\Re _\ell }} \right\}\nonumber\\
    &= {\tau _\ell }\mathbb{E}\left\{ {{{\log }_2}\left( {1 + {\kappa _\ell }\bar \gamma } \right) - {{\log }_2}\left( {1 + {\kappa _\ell }\left( {1 - {\alpha _c}} \right)\bar \gamma } \right)} \right\}\nonumber\\
    & + \mathbb{E}\left\{ {{{\log }_2}\left( {1 + \left( {{\alpha _\ell } + {\alpha _{\ell '}}} \right){\kappa _\ell }\bar \gamma } \right) - {{\log }_2}\left( {1 + {\kappa _\ell }{\alpha _{\ell '}}\bar \gamma } \right)} \right\}.
\end{align}

\begin{proposition}\label{proposition:3}
    We investigate the following function
    \begin{align}\label{eq:EC_1}
        {{Z}_\ell }\left( c \right) = \mathbb{E}\left\{ {{{\log }_2}\left( {1 + c{\kappa _\ell }} \right)} \right\},
    \end{align}
    where $c$ is an arbitrary non-negative constant and $\kappa_\ell$ is the considered random variable.
    Equation \eqref{eq:EC_1} is calculated as
    \begin{align*}
        {{Z}_\ell }\left( c \right) &= \frac{1}{{\ln \left( 2 \right)}}\int_0^\infty  {\frac{{1 - {F_{{\kappa _\ell }}}\left( {x/c} \right)}}{{1 + x}}dx}  \nonumber\\
        &\mathop  = \limits^{\left( a \right)} \frac{1}{{\ln \left( 2 \right)\Gamma \left( {{\xi _\ell }} \right)}}\int_0^\infty  {\frac{1}{{1 + x}}{\rm{G}}_{1,2}^{2,0}\left[ {\frac{{\sqrt x }}{{\Omega _\ell ^{}\sqrt c }}\left| {
        \begin{matrix}
            1\\{{\xi _\ell },0}
        \end{matrix} } \right.} \right]dx} \nonumber\\
        &\mathop  = \limits^{\left( b \right)} \frac{{{2^{{\xi _\ell } - 1}}}}{{\ln \left( 2 \right)\Gamma \left( {{\xi _\ell }} \right)\sqrt {\pi } }} {\rm{G}}_{3,5}^{5,1}\left[ {\frac{1}{{4c\Omega _\ell ^2}}\left| {
        \begin{matrix}
            {0,\,\,\frac{1}{2},\,\,1}\\
            {0,\,\frac{\xi _\ell }{2},\,0,\,\frac{ {{\xi _\ell } + 1} }{2},\,\frac{1}{2}}
        \end{matrix}
        } \right.} \right],
    \end{align*}
    where step $(a)$ is derived using
    ${F_{{\kappa _\ell }}}\left( x \right) = \frac{1}{{\Gamma \left( {{\xi _\ell }} \right)}}\Upsilon \left( {{\xi _\ell },\frac{{\sqrt x }}{{\Omega _\ell ^{}}}} \right) = 1 - \frac{1}{{\Gamma \left( {{\xi _\ell }} \right)}}\Gamma \left( {{\xi _\ell },\frac{{\sqrt x }}{{\Omega _\ell ^{}}}} \right)$ \cite{duong2012cognitive} and \cite[Eq.~(8.4.16.2)]{prudnikov1990integrals}, while step $(b)$ is derived using \cite[Eq.~(2.24.2.4)]{prudnikov1990integrals}.
\end{proposition}

From Proposition~\ref{proposition:3}, the closed-form expression for ${\rm{E}}{{\rm{C}}_\ell }$ in \eqref{eq:EC_l} can be expressed in terms of ${{Z}_\ell }\left( \cdot \right)$ such that
\begin{align}
    {\rm{E}}{{\rm{C}}_\ell } &= {\tau _\ell }\left\{ {{{Z}_\ell }\left( {\bar \gamma } \right) - {{Z}_\ell }\left( {\left( {1 - {\alpha _c}} \right)\bar \gamma } \right)} \right\} \nonumber\\
    &\qquad \qquad+ {{Z}_\ell }\left( {\left( {{\alpha _\ell } + {\alpha _{\ell '}}} \right)\bar \gamma } \right) - {{Z}_\ell }\left( {{\alpha _{\ell '}}\bar \gamma } \right)\nonumber\\
    &= {\tau _\ell } {Z}_\ell(\bar \gamma) - {{Z}_\ell }\left( {{\alpha _{\ell '}}\bar \gamma } \right) + (1-\tau_\ell){Z}_\ell((1-\alpha_{\rm{c}})\bar \gamma).
    \label{eq:EC_final}
\end{align}

\subsubsection{Asymptotic Ergodic Capacity}

In the high-SNR regime, we utilize the limit ${\log _2}\left( {1 + c\bar \gamma } \right) \to {\log _2}\left( {\bar \gamma } \right)$, as $\bar\gamma \to \infty$, for \eqref{eq:EC_l}.
Accordingly, the asymptotic EC is calculated as
\begin{align}\label{eq:asymptotic_EC}
    {\widetilde {{\rm{EC}}}_\ell } &=  - {\tau _\ell }{\log _2}\left( {1 - {\alpha _{\rm{c}}}} \right) + {\log _2}\left( {{\alpha _\ell } + {\alpha _{\ell '}}} \right) - {\log _2}\left( {{\alpha _{\ell '}}} \right)\nonumber\\
    & = (1-\tau_\ell) {\log _2}\left( {1 - {\alpha _{\rm{c}}}} \right) - {\log _2}\left( {{\alpha _{\ell '}}} \right).
\end{align}

\begin{remark}\label{remark:EC}
    The result in \eqref{eq:asymptotic_EC} reveals that the EC becomes saturated in the high-SNR regime, as ${\widetilde {{\rm{EC}}}_\ell }$ is independent of the average transmit SNR $\bar \gamma$.
    Furthermore, the asymptotic EC is sensitive to the PA and RS coefficients, as these parameters determine how power and rate resources are distributed between the common and private streams.
    Guaranteeing an RA that balances the capacity between users remains a crucial challenge for practical deployments.
\end{remark}

\subsection{Throughput-based and Spectral-based Energy Efficiencies}

EE is a crucial performance indicator that characterizes the tradeoff between transmission efficiency and energy consumption. It reflects how effectively the system converts its power budget into useful throughput or spectral gains.
This tradeoff becomes especially relevant in STAR-RIS-assisted RSMA IoT systems, where both transmission power and the active STAR-RIS amplification contribute to overall system efficiency.
To examine this balance, we evaluate two types of EE metrics: throughput-based EE and spectral-based EE, each emphasizing a different aspect of system performance.

Throughput-based EE measures the number of successfully delivered bits per unit of consumed power, explicitly accounting for the users’ outage probabilities. It is defined as
\begin{align}\label{eq:EE_TP}
    {\rm{E}}{{\rm{E}}_{{\rm{TP}}}} = \frac{{\sum\nolimits_\ell  {\left( {R_{\rm{c}}^{{\rm{th}}}{\tau _\ell } + R_\ell ^{{\rm{th}}}} \right)\left( {1 - {\rm{O}}{{\rm{P}}_\ell }} \right)} }}{{{P_0} + {P_S} + \psi \left( {\beta _{{\Sigma _2}}^{\rm{R}} + \beta _{{\Sigma _2}}^{\rm{T}}} \right)}},
\end{align}
where $P_0$ denotes the static circuit power,
$\psi = 0$ if $\delta _i^{\rm{R}} = \delta _i^{\rm{T}} = 1$,
and $\psi = 1$ if $\{ \delta _i^{\rm{R}},\delta _i^{\rm{T}}\} \ge 1$.

In contrast, spectral-based EE is based on the ergodic capacities of the common and private streams. It captures the long-term achievable spectral efficiency per unit of power consumption, and is expressed as
\begin{align}\label{eq:EE_SE}
    {\rm{E}}{{\rm{E}}_{{\rm{SE}}}} = \frac{{\sum\nolimits_\ell  {\left( {\mathbb{E}\left\{ {{\Re _{{\rm{c}},\ell }}} \right\} + \mathbb{E}\left\{ {{\Re _\ell }} \right\}} \right)} }}{{{P_0} + {P_S} + \psi \left( {\beta _{{\Sigma _2}}^{\rm{R}} + \beta _{{\Sigma _2}}^{\rm{T}}} \right)}}.
\end{align}

\begin{remark}
    Throughput-based EE is directly influenced by users’ OPs, making it suitable for reliability-critical applications where unsuccessful decoding must be accounted for. Spectral-based EE, however, depends on ergodic rates and is more suitable for assessing long-term performance under typical channel variations. Thus, the choice of EE metric depends on whether instantaneous reliability or long-term average efficiency is prioritized.
\end{remark}

\begin{remark}
    Both EE metrics are highly sensitive to the STAR-RIS configuration and the RSMA parameters.
    Increasing ${\beta _{{\Sigma _2}}^{\rm{R}}}$ and ${\beta _{{\Sigma _2}}^{\rm{T}}}$ enhances the beamforming gain but also raises the overall power consumption, creating a clear tradeoff. Likewise, PA and RS factors determine the balance between common and private stream rates, directly influencing the numerators of \eqref{eq:EE_TP} and \eqref{eq:EE_SE}.
    Consequently, careful optimization of these parameters is essential for maximizing the overall energy efficiency of the system.
\end{remark}

\begin{lemma}[Engineering design insights]\label{Lemma:DesignInsights}
The closed-form expressions translate into the following practical guidelines.
(\textit{i}) \emph{Reliability lever:} the diversity order is $D_\ell=\xi_\ell/2$ with $\xi_\ell=\mu_\ell^2/\chi_\ell$, which grows almost linearly with the number of elements $N$ and with the amplification gain $\delta$ through $\mu_\ell$ and $\chi_\ell$ in Proposition~\ref{proposition:1}. A target reliability can thus be met by trading hardware (more elements) against power (higher gain), giving a concrete rule for dimensioning $N$ and $\delta$.
(\textit{ii}) \emph{Coverage lever:} the array gain $A_\ell=G_\ell^{-2/\xi_\ell}$ sets the horizontal SNR offset of the OP curve and is tunable via the PA/RS coefficients and the target rates $R_{\rm{c}}^{\rm{th}},R_\ell^{\rm{th}}$; this directly links the quality-of-service (QoS) thresholds to the transmit power required to reach a given outage level.
(\textit{iii}) \emph{Fairness/efficiency lever:} the asymptotic EC in \eqref{eq:asymptotic_EC} is SNR-independent and fully determined by the PA/RS split, so at high SNR fairness is governed solely by resource allocation rather than transmit power; moreover, the EE metrics are unimodal in $\bar\gamma$, implying an optimal operating point beyond which added power is wasted, and the passive mode is preferable at low-to-moderate SNR. These insights motivate the fairness-oriented RA in Section~\ref{Sect:Model_based_RA}.
\end{lemma}
\begin{proof}
The three statements follow directly from Proposition~\ref{proposition:diversity_order} and Remark~\ref{remark:OP}, from the array-gain definition $A_\ell=G_\ell^{-2/\xi_\ell}$, and from \eqref{eq:asymptotic_EC} together with the EE expressions \eqref{eq:EE_TP}--\eqref{eq:EE_SE}, respectively.
\end{proof}

\section{Model-Based Resource Allocation Framework}
\label{Sect:Model_based_RA}

This section delves into two practical fairness optimization problems: user outage fairness and user capacity fairness.

\subsection{Outage Fairness Optimization}\label{Sect:OutageFairness}

\subsubsection{Problem Formulation}
Building on the OP analysis in Section~\ref{Sect:OP},
this part aims to improve fairness in outage performance among users.
In STAR-RIS-assisted RSMA IoT systems, near and far users often experience heterogeneous channel conditions; without proper RA, the far user may suffer disproportionately high outage levels.
To mitigate this imbalance, we jointly optimize the PA coefficients
${\bm{\alpha}} = \left\{ {{\alpha _{\rm{c}}},{\alpha _{\rm{n}}},{\alpha _{\rm{f}}}} \right\}$
and
the RS coefficients ${\bm{\tau }} = \left\{ {{\tau _{\rm{n}}},{\tau _{\rm{f}}}} \right\}$
to ensure a more equitable level of service reliability.
The objective is to minimize the worst-user OP while satisfying system constraints.
Mathematically, the outage fairness optimization problem is formulated as
\begin{subequations}\label{eq:Original_OP_Opt}
\begin{align}
    \label{eq:Original_OP_Opt_a}
    &\mathop {\min }\limits_{{\bm{\alpha }},{\bm{\tau }}} \mathop {\max }\limits_\ell  {\rm{O}}{{\rm{P}}_\ell }   \\
    \label{eq:Original_OP_Opt_b}
    \rm{s.t.} \quad
    &{\alpha _{\rm{c}}} + {\alpha _{\rm{n}}} + {\alpha _{\rm{f}}} \le 1, {\alpha _{\rm{c}}} \ge {\alpha _\ell },\forall \ell,\\
    \label{eq:Original_OP_Opt_d}
    &{\alpha _{\rm{c}}} - \varrho  \ge \gamma _{\rm{c},\ell }^{{\rm{th}}}\left( {1 - {\alpha _{\rm{c}}}} \right),\forall \ell ,\\
    \label{eq:Original_OP_Opt_e}
    & {\alpha _\ell } - \varrho  \ge \gamma _\ell ^{{\rm{th}}}{\alpha _{\ell '}},\forall \ell ,\ell ',\\
    \label{eq:Original_OP_Opt_f}
    &{\tau _{\rm{n}}} + {\tau _{\rm{f}}} \le 1,
\end{align}
\end{subequations}
where $\varrho = 10^{-3}$ is introduced to prevent division by zero and ensure numerical stability.

The optimization problem in \eqref{eq:Original_OP_Opt} is non-convex because of the max-min structure, the nonlinear coupling of variables, and the involvement of $\Upsilon(\cdot,\cdot)$ in OP expressions.
In the sequel, we attempt to approximate \eqref{eq:Original_OP_Opt} by a convex problem by utilizing an SCA method.

\subsubsection{Proposed Solution}
We first introduce an auxiliary variable $\theta$ such that
$\theta \ge \mathop {\max }\nolimits_\ell \{ {\rm{O}}{{\rm{P}}_\ell } \}$ and exploit the approximation $\rm{OP}_{\ell} \simeq {\widetilde {{\rm{OP}}}_\ell }$, which yields the surrogate objective function as $\mathop {\min }\nolimits_{{\bf{\alpha }},{\bf{\tau }},\theta } \theta $
subject to \cite{le2024active}
\begin{align}
    &\theta  \ge \mathop {\max }\nolimits_\ell  {\widetilde {{\rm{OP}}}_\ell } = \mathop {\max }\nolimits_\ell  \left\{ {\max {{\left( {{\wp _{{\rm{c}},\ell }},{\wp _\ell }} \right)}^{{\xi _\ell }/2}}/{{\bar \Omega }_\ell }} \right\}\nonumber\\
    \label{eq:Constraint21}
    &\mathop  \Rightarrow \limits^{\left( {{c}} \right)}
    \left\{ {{{\bar \Omega }_\ell }\theta  \ge \wp _{{\rm{c}},\ell }^{{\xi _\ell }/2}\, \wedge \,{{\bar \Omega }_\ell }\theta  \ge \wp _\ell ^{{\xi _\ell }/2}} \right\},
\end{align}
where ${{\bar \Omega }_\ell } = \Gamma \left( {{\xi _\ell } + 1} \right)\Omega _\ell ^{^{{\xi _\ell }}}$
and step $(c)$ is based on the fact that $x \ge \max(a,b)$ is equivalent to
$\{ x\ge a \wedge x\ge b \}$.
Subsequently, ${\wp _{{\rm{c}},\ell }^{{\xi _\ell }/2}}$
and
${\wp _\ell ^{{\xi _\ell }/2}}$ in
\eqref{eq:Constraint21} are convexified via first-order Taylor expansion at arbitrary nonzero points ${\bar \wp _{{\rm{c}},\ell }^{}}$
and
${\bar \wp _\ell ^{}}$ such that
\begin{align}\label{eq:33}
    \begin{cases}
        \wp _{{\rm{c}},\ell }^{{\xi _\ell }/2} \le \bar \wp _{{\rm{c}},\ell }^{{\xi _\ell }/2} + {c_{1,\ell }}\left( {\wp _{{\rm{c}},\ell }^{} - \bar \wp _{{\rm{c}},\ell }^{}} \right), \forall \ell,\\
        \wp _\ell ^{{\xi _\ell }/2} \le \bar \wp _\ell ^{{\xi _\ell }/2} + {c_{2,\ell }}\left( {\wp _\ell ^{} - \bar \wp _\ell ^{}} \right), \forall \ell,
    \end{cases}
\end{align}
where ${c_{1,\ell }} = \bar \wp _{{\rm{c}},\ell }^{{\xi _\ell }/2 - 1}{\xi _\ell }/2$
and
${c_{2,\ell }} = \bar \wp _\ell ^{{\xi _\ell }/2 - 1}{\xi _\ell }/2$.

Substituting \eqref{eq:33} into \eqref{eq:Constraint21} yields the equivalent constraints
\begin{align}\label{eq:34}
    \begin{cases}
        {{\bar c}_{1,\ell }} + \frac{{{c_{1,\ell }}\gamma _{{\rm{c}},\ell }^{{\rm{th}}}}}{{{\alpha _{\rm{c}}}\bar \gamma  - \gamma _{{\rm{c}},\ell }^{{\rm{th}}}\left( {1 - {\alpha _{\rm{c}}}} \right)\bar \gamma }} \le \bar \Omega _\ell ^{}\theta ,\forall \ell ,\\
        {{\bar c}_{2,\ell }} + \frac{{{c_{2,\ell }}\gamma _\ell ^{{\rm{th}}}}}{{{\alpha _\ell }\bar \gamma  - \gamma _\ell ^{{\rm{th}}}{\alpha _{\ell '}}\bar \gamma }} \le \bar \Omega _\ell ^{}\theta ,\forall \ell ,
    \end{cases}
\end{align}
where
${{\bar c}_{1,\ell }} = \bar \wp _{{\rm{c}},\ell }^{{\xi _\ell }/2} - {c_{1,\ell }}\bar \wp _{{\rm{c}},\ell }^{} = \left( {1 - {\xi _\ell }/2} \right)\bar \wp _{{\rm{c}},\ell }^{{\xi _\ell }/2}$
and
${{\bar c}_{2,\ell }} = \bar \wp _\ell ^{{\xi _\ell }/2} - {c_{2,\ell }}\bar \wp _\ell ^{} = \left( {1 - {\xi _\ell }/2} \right)\bar \wp _\ell ^{{\xi _\ell }/2}$.
We further introduce auxiliary variables $t_{1,\ell}$ and $t_{2,\ell}$ to bound
$\bar \Omega _\ell ^{}\theta  - {{\bar c}_{1,\ell }}$
and
$\bar \Omega _\ell ^{}\theta  - {{\bar c}_{2,\ell }}$ in \eqref{eq:34}, respectively.
The constraints in \eqref{eq:34} become
\begin{align}
\label{eq:35}
    \frac{1}{{{t_{1,\ell }}}}{c_{1,\ell }}\gamma _{{\rm{c}},\ell }^{{\rm{th}}} &\le {\alpha _{\rm{c}}}\bar \gamma \left( {1 + \gamma _{{\rm{c}},\ell }^{{\rm{th}}}} \right) - \gamma _{{\rm{c}},\ell }^{{\rm{th}}}\bar \gamma ,\forall \ell ,\\
\label{eq:36}
    {t_{1,\ell }} &\le \bar \Omega _\ell ^{}\theta  - {{\bar c}_{1,\ell }},\forall \ell ,\\
\label{eq:37}
    \frac{1}{{{t_{2,\ell }}}}{c_{2,\ell }}\gamma _\ell ^{{\rm{th}}} &\le {\alpha _\ell }\bar \gamma  - \gamma _\ell ^{{\rm{th}}}{\alpha _{\ell '}}\bar \gamma ,\forall \ell ,\\
\label{eq:38}
    {t_{2,\ell }} &\le \bar \Omega _\ell ^{}\theta  - {{\bar c}_{2,\ell }},\forall \ell.
\end{align}

Finally, the problem in \eqref{eq:Original_OP_Opt} can be equivalently recast as
\begin{align}   \label{eq:Equivalent_OP_Opt}
    \mathop {\min }\limits_{{\bm{\alpha }},{\bm{\tau }},\theta, {\left\{ {{t_{1,\ell }},{t_{2,\ell }}} \right\}_{\forall \ell }}}\, \theta
    \quad
    \rm{s.t.}
    \quad
    \eqref{eq:Original_OP_Opt_b}-\eqref{eq:Original_OP_Opt_f}, \eqref{eq:35}-\eqref{eq:38}.
\end{align}
Unfortunately, the problem \eqref{eq:Equivalent_OP_Opt} still remains non-convex due to highly coupled variables in \eqref{eq:Original_OP_Opt_d}, \eqref{eq:Original_OP_Opt_e}, and \eqref{eq:35}.
To tackle this challenge, we apply the BCD method to alternately solve \eqref{eq:Equivalent_OP_Opt} via two tractable sub-problems until the objective function converges.
The two sub-problems of the BCD algorithm are referred to as
\textbf{PA Block} and
\textbf{RS Block}.

\paragraph{\textbf{PA Block}}
Given the feasible points $\bar \tau_{\rm{n}}$, $\bar \tau_{\rm{f}}$,
and $\bar \gamma _{{\rm{c}},\ell }^{{\rm{th}}} = {2^{R_{\rm{c}}^{{\rm{th}}}/{{\bar \tau }_\ell }}} - 1$, the problem with respect to the PA coefficients ${\bm{\alpha}} = \left\{ {{\alpha _{\rm{c}}},{\alpha _{\rm{n}}},{\alpha _{\rm{f}}}} \right\}$ is expressed as
\begin{subequations}\label{eq:PA_block}
\begin{align}
    \label{eq:PA_block_a}
    &\mathop {\min }\limits_{{\bm{\alpha }},\theta, {\left\{ {{t_{1,\ell }},{t_{2,\ell }}} \right\}_{\forall \ell }}}\, \theta    \\
    \label{eq:PA_block_b}
    \rm{s.t.} \quad
    & \eqref{eq:Original_OP_Opt_b}, \eqref{eq:Original_OP_Opt_e}, \eqref{eq:36}-\eqref{eq:38},\\
    \label{eq:PA_block_c}
    &{\alpha _{\rm{c}}} - \varrho  \ge \bar \gamma _{\rm{c},\ell }^{{\rm{th}}}\left( {1 - {\alpha _{\rm{c}}}} \right),\forall \ell , \\
    \label{eq:PA_block_d}
    &\frac{1}{{{t_{1,\ell }}}}{c_{1,\ell }}\bar\gamma _{{\rm{c}},\ell }^{{\rm{th}}} \le {\alpha _{\rm{c}}}\bar \gamma \left( {1 + \bar\gamma _{{\rm{c}},\ell }^{{\rm{th}}}} \right) - \bar\gamma _{{\rm{c}},\ell }^{{\rm{th}}}\bar \gamma ,\forall \ell ,
\end{align}
\end{subequations}
where constraints \eqref{eq:PA_block_c} and \eqref{eq:PA_block_d} are manipulated from constraints \eqref{eq:Original_OP_Opt_d} and \eqref{eq:35}, respectively.
Now, \eqref{eq:PA_block} becomes a convex optimization problem, which can be solved using CVX with a semidefinite programming (SDP) solver.

\paragraph{\textbf{RS Block}}

Given the feasible points ${{\bar \alpha }_{\rm{c}}}$,
${{\bar \alpha }_{\rm{n}}}$,
and ${{\bar \alpha }_{\rm{f}}}$,
the optimization problem with respect to $\bm{\tau}$ is expressed as
\begin{subequations}\label{eq:RS_block}
\begin{align}
    \label{eq:RS_block_a}
    &\mathop {\min }\limits_{{\bm{\tau }},\theta, {\left\{ {{t_{1,\ell }},{t_{2,\ell }}} \right\}_{\forall \ell }}}\, \theta    \\
    \label{eq:RS_block_b}
    \rm{s.t.} \quad
    & \eqref{eq:Original_OP_Opt_f}, \eqref{eq:36}, \eqref{eq:38},\\
    \label{eq:RS_block_c}
    &{\tau _\ell } \ge R_{\rm{c}}^{{\rm{th}}}\ln \left( 2 \right)/\ln \left( {1 + \frac{{{{\bar \alpha }_{\rm{c}}} - \varrho }}{{1 - {{\bar \alpha }_{\rm{c}}}}}} \right),\forall \ell , \\
    \label{eq:RS_block_d}
    &\frac{R_{\rm{c}}^{{\rm{th}}}\ln \left( 2 \right)}{{{\tau _\ell }}} \le \ln \big( {1 + \frac{{{t_{1,\ell }}{{\bar \alpha }_{\rm{c}}}\bar \gamma }}{{{c_{1,\ell }} + {t_{1,\ell }}\left( {1 - {{\bar \alpha }_{\rm{c}}}} \right)\bar \gamma }}} \big),\forall \ell ,\\
    \label{eq:RS_block_e}
    &\frac{1}{{{t_{2,\ell }}}}{c_{2,\ell }}\gamma _\ell ^{{\rm{th}}} \le {{\bar \alpha }_\ell }\bar \gamma  - \gamma _\ell ^{{\rm{th}}}{{\bar \alpha }_{\ell '}}\bar \gamma ,\forall \ell ,\
\end{align}
\end{subequations}
where constraints \eqref{eq:RS_block_c}, \eqref{eq:RS_block_d}, and \eqref{eq:RS_block_e} are manipulated from constraints
\eqref{eq:Original_OP_Opt_d}, \eqref{eq:35}, and \eqref{eq:37}, respectively.
Although \eqref{eq:RS_block} has been simplified, it remains non-convex due to the non-convexity of \eqref{eq:RS_block_d}.
To convexify \eqref{eq:RS_block_d}, we apply the inequality in \cite[Eq.~(34)]{vu2024enhancing} to bound $f\left( {{t_{1,\ell }}} \right) \buildrel \Delta \over =  \ln \left( {1 + \frac{{{t_{1,\ell }}{{\bar \alpha }_{\rm{c}}}\bar \gamma }}{{{c_{1,\ell }} + {t_{1,\ell }}\left( {1 - {{\bar \alpha }_{\rm{c}}}} \right)\bar \gamma }}} \right)$ as
$f\left( {{t_{1,\ell }}} \right) \ge f\left( {{{\bar t}_{1,\ell }}} \right) + \left( {2 - \frac{{{{\bar t}_{1,\ell }}}}{{{t_{1,\ell }}}} - \frac{{{c_{1,\ell }} + {t_{1,\ell }}\left( {1 - {{\bar \alpha }_{\rm{c}}}} \right)\bar \gamma }}{{{c_{1,\ell }} + {{\bar t}_{1,\ell }}\left( {1 - {{\bar \alpha }_{\rm{c}}}} \right)\bar \gamma }}} \right)
\times {\left( {1 + \frac{{{c_{1,\ell }} + {{\bar t}_{1,\ell }}\left( {1 - {{\bar \alpha }_{\rm{c}}}} \right)\bar \gamma }}{{{{\bar t}_{1,\ell }}{{\bar \alpha }_{\rm{c}}}\bar \gamma }}} \right)^{ - 1}}, \forall \ell$,
where ${{{\bar t}_{1,\ell }}}$ is the feasible point obtained from the previous iteration.
Accordingly, \eqref{eq:RS_block_d} can be rewritten as
\begin{align}\label{eq:43}
    &\frac{{{{\bar t}_{1,\ell }}{{\bar \alpha }_{\rm{c}}}\bar \gamma }}{{{c_{1,\ell }} + {{\bar t}_{1,\ell }}\bar \gamma }} 
    \left( {\frac{{{c_{1,\ell }} + 2{{\bar t}_{1,\ell }}\left( {1 - {{\bar \alpha }_{\rm{c}}}} \right)\bar \gamma  - {t_{1,\ell }}\left( {1 - {{\bar \alpha }_{\rm{c}}}} \right)\bar \gamma }}{{{c_{1,\ell }} + {{\bar t}_{1,\ell }}\left( {1 - {{\bar \alpha }_{\rm{c}}}} \right)\bar \gamma }}} \right)\nonumber\\
    &\quad + f\left( {{{\bar t}_{1,\ell }}} \right) 
    \ge \frac{1}{{{t_{1,\ell }}}}\frac{{\bar t_{1,\ell }^2{{\bar \alpha }_{\rm{c}}}\bar \gamma }}{{\left( {{c_{1,\ell }} + {{\bar t}_{1,\ell }}\bar \gamma } \right)}} + \frac{1}{{{\tau _\ell }}}R_{\rm{c}}^{{\rm{th}}}\ln \left( 2 \right),\forall \ell.
\end{align}
As a result, the problem in \eqref{eq:RS_block} is equivalently recast as
\begin{align}
    \label{eq:RS_block2}
    \mathop {\min }\limits_{{\bm{\tau }},\theta, {\left\{ {{t_{1,\ell }},{t_{2,\ell }}} \right\}_{\forall \ell }}}\, \theta    \,\,\,
    \rm{s.t.} \,\,
    \eqref{eq:Original_OP_Opt_f}, \eqref{eq:36}, \eqref{eq:38}, \eqref{eq:RS_block_c}, \eqref{eq:RS_block_e}, \eqref{eq:43}.
\end{align}
Now, \eqref{eq:RS_block2} becomes a convex optimization problem, which can be solved using CVX with an SDP solver.

\subsubsection{Overall Algorithm, Convergence, and Complexity}
\label{Sect:4.1.3}

\begin{algorithm}[!t]
    \small
    \caption{{Outage Fairness Algorithm}}
    \label{Algo:Overall_algorithm}
    \begin{algorithmic}[1]
        \State Initialize iteration number $i=0$, threshold $\epsilon = 10^{-4}$, and $\left( {\alpha _{\rm{c}}^{\left( 0 \right)},\alpha _{\rm{n}}^{\left( 0 \right)},\alpha _{\rm{f}}^{\left( 0 \right)},\tau _{\rm{n}}^{\left( 0 \right)},\tau _{\rm{f}}^{\left( 0 \right)}} \right)$ in the feasible domain;
        \State Calculate
        $\bar \wp _\ell ^{} = \gamma _\ell ^{{\rm{th}}}{\left[ {\alpha _\ell ^{\left( 0 \right)}\bar \gamma  - \gamma _\ell ^{{\rm{th}}}\alpha _{\ell '}^{\left( 0 \right)}\bar \gamma } \right]^{ - 1}}$,
        $\bar \wp _{{\rm{c}},\ell }^{} = \left( {{2^{R_{\rm{c}}^{{\rm{th}}}/\tau _\ell ^{\left( 0 \right)}}} - 1} \right){\left[ {\alpha _{\rm{c}}^{\left( 0 \right)}\bar \gamma  - \left( {{2^{R_{\rm{c}}^{{\rm{th}}}/\tau _\ell ^{\left( 0 \right)}}} - 1} \right)\left( {1 - \alpha _{\rm{c}}^{\left( 0 \right)}} \right)\bar \gamma } \right]^{ - 1}}$,
        ${\theta ^{\left( 0 \right)}} = {\max _\ell }\left\{ {\max {{\left( {\bar \wp _{{\rm{c}},\ell }^{},\bar \wp _\ell ^{}} \right)}^{{\xi _\ell }/2}}/\bar \Omega _\ell ^{}} \right\}$,
        $\bar t_{1,\ell }^{} = \bar \Omega _\ell ^{}{\theta ^{\left( 0 \right)}} - \left( {1 - {\xi _\ell }/2} \right){\left( {\bar \wp _{{\rm{c}},\ell }^{}} \right)^{{\xi _\ell }/2}}$,
        ${{\bar \alpha }_{\rm{c}}} = \alpha _{\rm{c}}^{\left( 0 \right)}$, ${{\bar \alpha }_{\rm{n}}} = \alpha _{\rm{n}}^{\left( 0 \right)}$, and ${{\bar \alpha }_{\rm{f}}} = \alpha _{\rm{f}}^{\left( 0 \right)}$,
        $\forall \ell$;
        \Repeat
            \State Update iterative index $i \leftarrow i+1$;
            \State Solve the RS block \eqref{eq:RS_block2} using CVX with an SDP solver and obtain $\tau_\ell^{(i)}$, $\forall \ell$;
            \State Update
            $\bar \tau _{\rm{n}}^{} \leftarrow \tau _{\rm{n}}^{\left( i \right)}$,
            $\bar \tau _{\rm{f}}^{} \leftarrow \tau _{\rm{f}}^{\left( i \right)}$,
            $\bar \gamma _{{\rm{c}},\ell }^{{\rm{th}}} \leftarrow {2^{R_{\rm{c}}^{{\rm{th}}}/\tau _\ell ^{\left( i \right)}}} - 1$,
            and
            $\bar \wp _{{\rm{c}},\ell }^{} \leftarrow \bar \gamma _{{\rm{c}},\ell }^{{\rm{th}}}{\left[ {{{\bar \alpha }_{\rm{c}}}\bar \gamma  - \bar \gamma _{{\rm{c}},\ell }^{{\rm{th}}}\left( {1 - {{\bar \alpha }_{\rm{c}}}} \right)\bar \gamma } \right]^{ - 1}}$, $\forall \ell$;
            \State Solve the PA block \eqref{eq:PA_block} using CVX with an SDP solver and obtain
            $\alpha _{\rm{c}}^{\left( i \right)}$,
            $\alpha _{\rm{n}}^{\left( i \right)}$,
            $\alpha _{\rm{f}}^{\left( i \right)}$,
            and $\theta^{(i)}$, $\forall \ell$;
            \State Update
            $\bar t_{1,\ell }^{} \leftarrow \bar \Omega _\ell ^{}{\theta ^{\left( i \right)}} - \left( {1 - {\xi _\ell }/2} \right){\left( {\bar \wp _{{\rm{c}},\ell }^{}} \right)^{{\xi _\ell }/2}}$,
            ${{\bar \alpha }_{\rm{c}}} \leftarrow \alpha _{\rm{c}}^{\left( i \right)}$,
            ${{\bar \alpha }_{\rm{n}}} \leftarrow \alpha _{\rm{n}}^{\left( i \right)}$,
            and ${{\bar \alpha }_{\rm{f}}} \leftarrow \alpha _{\rm{f}}^{\left( i \right)}$, $\forall \ell$;
        \Until $\left| {{\theta ^{\left( i \right)}} - {\theta ^{\left( {i - 1} \right)}}} \right|/{\theta ^{\left( {i - 1} \right)}} \le \epsilon$;
        \State Return optimal points
        $ \alpha _{\rm{c}}^* = \alpha _{\rm{c}}^{\left( i \right)}$,
        $ \alpha _{\rm{n}}^* = \alpha _{\rm{n}}^{\left( i \right)}$,
        $ \alpha _{\rm{f}}^* = \alpha _{\rm{f}}^{\left( i \right)}$,
        $\tau _{\rm{n}}^* = \tau _{\rm{n}}^{\left( i \right)}$,
        and $\tau _{\rm{f}}^* = \tau _{\rm{f}}^{\left( i \right)}$.
    \end{algorithmic}
\end{algorithm}

\paragraph{Outage Fairness Algorithm}
The overall outage fairness algorithm for RA is summarized in Algorithm~\ref{Algo:Overall_algorithm}.
The algorithm begins by initializing all variables within their feasible regions and then evaluating the objective function.
In each iteration, the RS block \eqref{eq:RS_block2} and the PA block \eqref{eq:PA_block} are alternately solved to update the corresponding variables.
This iterative process continues until the objective function converges. Upon convergence, the final RS and PA coefficients are returned, as shown in Step $10$ of Algorithm~\ref{Algo:Overall_algorithm}.

\paragraph{Convergence Analysis}
Considering iteration $i$ and the objective function as a function of RS and PA coefficients, i.e., ${\theta}\left( {\alpha _{\rm{c}}^{\left( i \right)},\alpha _{\rm{n}}^{\left( i \right)},\alpha _{\rm{f}}^{\left( i \right)},\tau _{\rm{n}}^{\left( i \right)},\tau _{\rm{f}}^{\left( i \right)}} \right)$, we have a series of inequalities
\begin{align*}
&\theta\! \left(\! {\alpha _{\rm{c}}^{\left( i \right)},\alpha _{\rm{n}}^{\left( i \right)},\alpha _{\rm{f}}^{\left( i \right)},\tau _{\rm{n}}^{\left( i \right)},\tau _{\rm{f}}^{\left( i \right)}} \!\right)\mathop  
\le \limits^{\left( d \right)} \theta\! \left( \!{\alpha _{\rm{c}}^{\left( i \right)},\alpha _{\rm{n}}^{\left( i \right)},\alpha _{\rm{f}}^{\left( i \right)},\tau _{\rm{n}}^{\left( {i + 1} \right)},\tau _{\rm{f}}^{\left( {i + 1} \right)}} \!\right)
\end{align*}
$\mathop  \le \limits^{\left( e \right)} \theta \left( {\alpha _{\rm{c}}^{\left( {i + 1} \right)},\alpha _{\rm{n}}^{\left( {i + 1} \right)},\alpha _{\rm{f}}^{\left( {i + 1} \right)},\tau _{\rm{n}}^{\left( {i + 1} \right)},\tau _{\rm{f}}^{\left( {i + 1} \right)}} \right)$,
where
steps $(d)$ and $(e)$ hold by the marginal optimality
of alternately solving \eqref{eq:RS_block2} and \eqref{eq:PA_block}, respectively.
When $\theta \big( {\alpha _{\rm{c}}^{\left( {i + 1} \right)},\alpha _{\rm{n}}^{\left( {i + 1} \right)},\alpha _{\rm{f}}^{\left( {i + 1} \right)},\tau _{\rm{n}}^{\left( {i + 1} \right)},\tau _{\rm{f}}^{\left( {i + 1} \right)}} \big)$
is bounded below due to constraints \eqref{eq:PA_block_c}, \eqref{eq:RS_block_c}, and \eqref{eq:RS_block_e},
the minimization of $\theta \big( {\alpha _{\rm{c}}^{\left( {i + 1} \right)},\alpha _{\rm{n}}^{\left( {i + 1} \right)},\alpha _{\rm{f}}^{\left( {i + 1} \right)},\tau _{\rm{n}}^{\left( {i + 1} \right)},\tau _{\rm{f}}^{\left( {i + 1} \right)}} \big)$ is convergent \cite{tu2025semi,tu2025hybrid}.

\paragraph{Computational Complexity}
In this work, the computational complexity is measured based on floating-point operations (FLOPs) \cite{hunger2005floating}.
Solving the RS block \eqref{eq:RS_block2} using an interior-point method \cite{luo2010semidefinite} requires ${\cal C}_1 = {\cal O}(n_1^{3.5} \ln(1/\epsilon_1) )$ FLOPs,
where $n_1 = 7$ represents the number of variables
in the CVX solver and $\epsilon_1 = 10^{-6}$ denotes the convergence accuracy of the RS block.
Likewise, the number of FLOPs to perform the PA block \eqref{eq:PA_block} is ${\cal C}_2 = {\cal O}(n_2^{3.5} \ln(1/\epsilon_2) )$,
where
$n_2 = 8$ represents the number of variables
in the CVX solver and $\epsilon_2 = 10^{-6}$ denotes the convergence accuracy of the PA block.
With $I_0$ denoting the number of outer iterations accumulated in Algorithm~\ref{Algo:Overall_algorithm}, the total complexity of the outage fairness algorithm is ${\cal O} (I_0({\cal C}_1 + {\cal C}_2))$.

\subsection{Capacity Fairness Optimization}\label{Sect:CapacityFairness}

\subsubsection{Problem Formulation}

Beyond outage performance, fairness can also be evaluated in terms of the achievable EC of each user. In STAR-RIS-aided RSMA IoT systems, the near user typically attains a higher capacity due to its more favorable channel conditions. To ensure that both users benefit proportionally from the available spectrum resources, we jointly optimize the PA coefficients
${\bm{\alpha}} = \left\{ {{\alpha _{\rm{c}}},{\alpha _{\rm{n}}},{\alpha _{\rm{f}}}} \right\}$
and
the RS coefficients ${\bm{\tau }} = \left\{ {{\tau _{\rm{n}}},{\tau _{\rm{f}}}} \right\}$ to balance the long-term achievable rates.
Capacity fairness aims to reduce the EC gap between users by solving
\begin{subequations}\label{eq:Original_EC_Opt}
\begin{align}
    \label{eq:Original_EC_Opt_a}
    &\mathop {\max }\limits_{{\bm{\alpha }},{\bm{\tau }}} \mathop {\min }\limits_\ell  {\rm{E}}{{\rm{C}}_\ell }   \\
    \label{eq:Original_EC_Opt_b}
    \rm{s.t.} \quad
    &{\tau _{\rm{n}}} + {\tau _{\rm{f}}} \le 1, {\alpha _{\rm{c}}} + {\alpha _{\rm{n}}} + {\alpha _{\rm{f}}} \le 1, {\alpha _{\rm{c}}} \ge {\alpha _\ell },\forall \ell,\\
    \label{eq:Original_EC_Opt_e}
    & \mathbb{E}\left\{ {{\Re _{{\rm{c}},\ell }}} \right\} \ge R_{\rm{c}}^{{\rm{th}}},\mathbb{E}\left\{ {{\Re _\ell }} \right\} \ge R_\ell ^{{\rm{th}}},\forall \ell,
\end{align}
\end{subequations}
where constraints \eqref{eq:Original_EC_Opt_e}
ensure QoS by guaranteeing that the long-term achievable rates for decoding the common stream and the private stream meet the required thresholds $R_{\rm{c}}^{{\rm{th}}}$ and $R_\ell ^{{\rm{th}}}$, respectively.

The problem in \eqref{eq:Original_EC_Opt} is non-convex due to the non-convexity of the objective function \eqref{eq:Original_EC_Opt_a} and constraints \eqref{eq:Original_EC_Opt_e}.
Similar to the outage fairness optimization, we convert \eqref{eq:Original_EC_Opt} into a convex problem via an SCA-based approach.

\subsubsection{Proposed Solution}
We begin by introducing an auxiliary variable satisfying $\omega  \le {\min _\ell } \{ {\rm{E}}{{\rm{C}}_\ell } \}$, which yields
\begin{align}\label{eq:46}
    \mathbb{E}\left\{ {{\Re _{{\rm{c}},\ell }}} \right\} + \mathbb{E}\left\{ {{\Re _\ell }} \right\} \ge \omega ,\forall \ell.
\end{align}
Thus, \eqref{eq:Original_EC_Opt} can be reformulated as
\begin{align}
    \label{eq:EC_Opt1}
    \mathop {\max }\limits_{{\bm{\alpha }},{\bm{\tau }},\omega}\, \omega    \quad
    \rm{s.t.} \quad
    \eqref{eq:Original_EC_Opt_b}-\eqref{eq:Original_EC_Opt_e}, \eqref{eq:46}.
\end{align}
The main source of non-convexity in \eqref{eq:EC_Opt1} lies in the logarithmic expressions in constraints \eqref{eq:Original_EC_Opt_e} and \eqref{eq:46}.
To convexify these terms, we proceed as follows.

For the convexification of $\mathbb{E}\left\{ {{\Re _{{\rm{c}},\ell }}} \right\}$, we have
\begin{align}\label{eq:48}
    \mathbb{E}\left\{ {{\Re _{{\rm{c}},\ell }}}{\left( {{\tau _\ell },{\alpha _{\rm{c}}}} \right)} \right\}
    \mathop  \ge \limits^{\left( f \right)} \mathbb{E}\left\{ {{{\hat \Re }_{{\rm{c}},\ell }}}{\left( {{\tau _\ell },{\alpha _{\rm{c}}}} \right)} \right\}
    \mathop  \ge \limits^{\left( g \right)} {{\tilde \Re }_{{\rm{c}},\ell }}\left( {{\tau _\ell },{\alpha _{\rm{c}}}} \right),
\end{align}
where steps $(f)$ and $(g)$ are based on \cite[Lemma~1]{nasir2019uav} and Jensen’s inequality, respectively.
Specifically, we utilize the inequality in \cite[Lemma~1]{nasir2019uav} to approximate ${{\Re _{{\rm{c}},\ell }}}\left( {{\tau _\ell },{\alpha _{\rm{c}}}} \right)$ at feasible points $\left( {{{\hat \tau }_\ell },{{\hat \alpha }_{\rm{c}}}} \right)$, which yields
${\Re _{{\rm{c}},\ell }}\left( {{\tau _\ell },{\alpha _{\rm{c}}}} \right)\mathop  \ge \limits^{\left( g \right)} \frac{{{{\hat \tau }_\ell }/\ln \left( 2 \right)}}{{1 + \frac{{{\kappa _\ell }\left( {1 - {{\hat \alpha }_{\rm{c}}}} \right)\bar \gamma  + 1}}{{{\kappa _\ell }{{\hat \alpha }_{\rm{c}}}\bar \gamma }}}}\left( {1 - \frac{{{{\hat \alpha }_{\rm{c}}}}}{{{\alpha _{\rm{c}}}}} - \frac{{\left( {2 - {\alpha _{\rm{c}}}} \right){\kappa _\ell }\bar \gamma  + 1}}{{\left( {1 - {{\hat \alpha }_{\rm{c}}}} \right){\kappa _\ell }\bar \gamma  + 1}}} \right) + 2{\Re _{{\rm{c}},\ell }}\left( {{{\hat \tau }_\ell },{{\hat \alpha }_{\rm{c}}}} \right) - \frac{{{{\hat \tau }_\ell }}}{{{\tau _\ell }}}{\Re _{{\rm{c}},\ell }}\left( {{{\hat \tau }_\ell },{{\hat \alpha }_{\rm{c}}}} \right) \buildrel \Delta \over = {{\hat \Re }_{{\rm{c}},\ell }}\left( {{\tau _\ell },{\alpha _{\rm{c}}}} \right)$.
Subsequently, based on Jensen's inequality, we establish that
$\mathbb{E}\left\{ {{{\hat \Re }_{{\rm{c}},\ell }}\left( {{\tau _\ell },{\alpha _{\rm{c}}}} \right)} \right\}
\mathop  \ge \limits^{\left( h \right)} \frac{{{{\hat \tau }_\ell }/\ln \left( 2 \right)}}{{1 + \frac{{{{\bar \kappa }_\ell }\left( {1 - {{\hat \alpha }_{\rm{c}}}} \right)\bar \gamma  + 1}}{{{{\bar \kappa }_\ell }{{\hat \alpha }_{\rm{c}}}\bar \gamma }}}}\left( {2 - \frac{{{{\hat \alpha }_{\rm{c}}}}}{{{\alpha _{\rm{c}}}}} - \frac{{\left( {1 - {\alpha _{\rm{c}}}} \right){{\bar \kappa }_\ell }\bar \gamma  + 1}}{{\left( {1 - {{\hat \alpha }_{\rm{c}}}} \right){{\bar \kappa }_\ell }\bar \gamma  + 1}}} \right)
+ 2{{\bar \Re }_{{\rm{c}},\ell }}\left( {{{\hat \tau }_\ell },{{\hat \alpha }_{\rm{c}}}} \right) - \frac{{{{\hat \tau }_\ell }}}{{{\tau _\ell }}}{{\bar \Re }_{{\rm{c}},\ell }}\left( {{{\hat \tau }_\ell },{{\hat \alpha }_{\rm{c}}}} \right) \buildrel \Delta \over = {{\tilde \Re }_{{\rm{c}},\ell }}\left( {{\tau _\ell },{\alpha _{\rm{c}}}} \right)$,
where
${{\bar \Re }_{{\rm{c}},\ell }}\left( {{{\hat \tau }_\ell },{{\hat \alpha }_{\rm{c}}}} \right) = {{\hat \tau }_\ell }{\log _2}\left( {1 + \frac{{{{\bar \kappa }_\ell }{{\hat \alpha }_{\rm{c}}}\bar \gamma }}{{{{\bar \kappa }_\ell }(1 - {{\hat \alpha }_{\rm{c}}})\bar \gamma  + 1}}} \right)$
and
${{\bar \kappa }_\ell } = \mathbb{E}\left\{ {{\kappa _\ell }} \right\} = \Omega _\ell ^2\Gamma \left( {{\xi _\ell } + 2} \right)/\Gamma \left( {{\xi _\ell }} \right)$.

For the convexification of $\mathbb{E}\left\{ {{\Re _\ell }} \right\}$, we have
\begin{align}\label{eq:51}
    \mathbb{E}\left\{ {{\Re _\ell }\left( {{\alpha _\ell },{\alpha _{\ell '}}} \right)} \right\}
    \mathop  \ge \limits^{\left( h \right)} \mathbb{E}\left\{ {{{\hat \Re }_\ell }\left( {{\alpha _\ell },{\alpha _{\ell '}}} \right)} \right\}
    \mathop  \ge \limits^{\left( i \right)} {{\tilde \Re }_\ell }\left( {{\alpha _\ell },{\alpha _{\ell '}}} \right),
\end{align}
where steps $(h)$ and $(i)$ are based on \cite[Eq.~(34)]{vu2024enhancing} and Jensen’s inequality, respectively.
Specifically, using the inequality in \cite[Eq.~(34)]{vu2024enhancing} for ${\Re _\ell }\left( {{\alpha _\ell },{\alpha _{\ell '}}} \right)$ at feasible points $\left( {{{\hat \alpha }_\ell },{{\hat \alpha }_{\ell '}}} \right)$ yields
\begin{align*}
    &{\Re _\ell }\left( {{\alpha _\ell },{\alpha _{\ell '}}} \right)\mathop  \ge \limits^{\left( h \right)} \frac{{{{\hat \tau }_\ell }/\ln \left( 2 \right)}}{{1 + \frac{{{\kappa _\ell }{{\hat \alpha }_{\ell '}}\bar \gamma  + 1}}{{{\kappa _\ell }{\alpha _\ell }\bar \gamma }}}}\left( {2 - \frac{{{{\hat \alpha }_\ell }}}{{{\alpha _\ell }}} - \frac{{{\kappa _\ell }{\alpha _{\ell '}}\bar \gamma  + 1}}{{{\kappa _\ell }{{\hat \alpha }_{\ell '}}\bar \gamma  + 1}}} \right)\\
    &\quad + {\Re _\ell }\left( {{{\hat \alpha }_\ell },{{\hat \alpha }_{\ell '}}} \right) \buildrel \Delta \over = {{\hat \Re }_\ell }\left( {{\alpha _\ell },{\alpha _{\ell '}}} \right).
\end{align*}
Subsequently, based on Jensen's inequality, we have
\begin{align*}
    &\mathbb{E}\left\{ {{{\hat \Re }_\ell }\left( {{\alpha _\ell },{\alpha _{\ell '}}} \right)} \right\}\mathop  \ge \limits^{\left( i \right)} \frac{{{{\hat \tau }_\ell }/\ln \left( 2 \right)}}{{1 + \frac{{{{\bar \kappa }_\ell }{{\hat \alpha }_{\ell '}}\bar \gamma  + 1}}{{{{\bar \kappa }_\ell }{\alpha _\ell }\bar \gamma }}}}\left( {2 - \frac{{{{\hat \alpha }_\ell }}}{{{\alpha _\ell }}} - \frac{{{{\bar \kappa }_\ell }{\alpha _{\ell '}}\bar \gamma  + 1}}{{{{\bar \kappa }_\ell }{{\hat \alpha }_{\ell '}}\bar \gamma  + 1}}} \right)\\
    &\quad + {\log _2}\left( {1 + \frac{{{{\bar \kappa }_\ell }{{\hat \alpha }_\ell }\bar \gamma }}{{{{\bar \kappa }_\ell }{{\hat \alpha }_{\ell '}}\bar \gamma  + 1}}} \right) \buildrel \Delta \over = {{\tilde \Re }_\ell }\left( {{\alpha _\ell },{\alpha _{\ell '}}} \right).
\end{align*}

From \eqref{eq:48} and \eqref{eq:51}, the problem in \eqref{eq:EC_Opt1} can be equivalently reformulated as
\begin{subequations}\label{eq:EC_Opt2}
\begin{align}
    \label{eq:EC_Opt2_a}
    &\mathop {\max }\limits_{{\bm{\alpha }},{\bm{\tau }},\omega}\, \omega   \\
    \label{eq:EC_Opt2_b}
    \rm{s.t.} \quad
    & \eqref{eq:Original_EC_Opt_b}, {{\tilde \Re }_{{\rm{c}},\ell }}\left( {{\tau _\ell },{\alpha _{\rm{c}}}} \right) + {{\tilde \Re }_\ell }\left( {{\alpha _\ell },{\alpha _{\ell '}}} \right) \ge \omega ,\forall \ell,\\
    \label{eq:EC_Opt2_c}
    &{{\tilde \Re }_{{\rm{c}},\ell }}\left( {{\tau _\ell },{\alpha _{\rm{c}}}} \right) \ge R_{\rm{c}}^{{\rm{th}}}, {{\tilde \Re }_\ell }\left( {{\alpha _\ell },{\alpha _{\ell '}}} \right) \ge R_\ell ^{{\rm{th}}},\forall \ell,
\end{align}
\end{subequations}
Now, \eqref{eq:EC_Opt2} becomes a convex optimization problem that can be solved using CVX with an SDP solver.
When \eqref{eq:EC_Opt2} is a maximization problem bounded above by constraints \eqref{eq:Original_EC_Opt_b}, the solver is guaranteed to converge.
The computational complexity of solving the capacity-fairness optimization is ${\cal C}_3 = {\cal O}(n_3^{3.5} \ln(1/\epsilon_3) )$ FLOPs,
where
$n_3 = 6$ represents the number of variables
in the CVX solver and $\epsilon_3 = 10^{-6}$ denotes the convergence tolerance.

\begin{figure*}[!t]
\centering
\includegraphics[width=\linewidth]{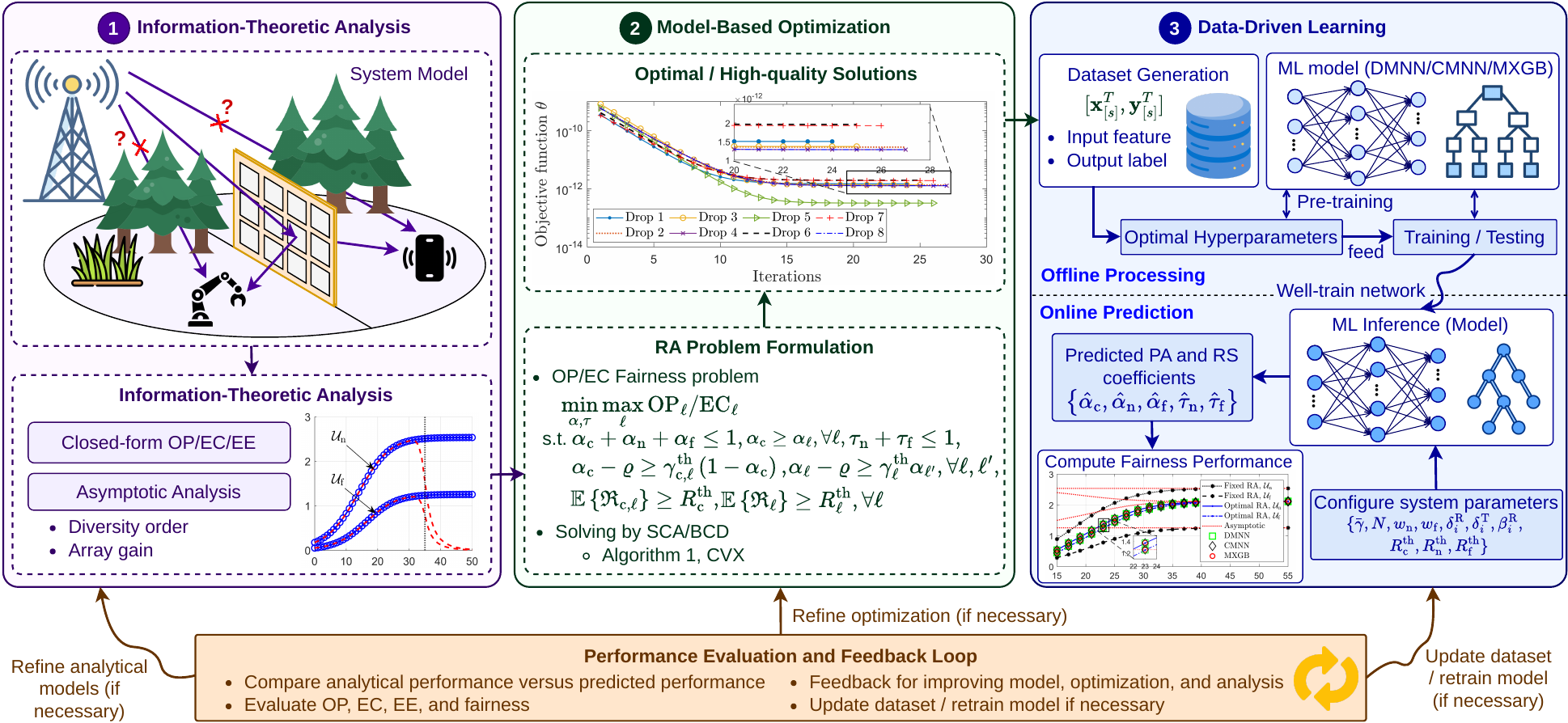}
\caption{Unified workflow illustrating the interaction among the information-theoretic analysis (Section~\ref{Sect:Analytical}), the model-based optimization (Section~\ref{Sect:Model_based_RA}), and the data-driven learning (Section~\ref{Sect:Data-Driven}) frameworks for fairness-oriented RA in active STAR-RIS-assisted RSMA IoT systems.}
\label{Fig:ML_App}
\end{figure*}

\begin{remark}\label{Remark:CapacityFairness_K}
The capacity-fairness problem \eqref{eq:Original_EC_Opt} generalizes to $K$ users along the grouping structure of Remark~\ref{Remark:MultiUser}, with objective $\max_{{\bm{\alpha}},{\bm{\tau}}} \min_{k} {\rm{EC}}_k$ subject to the PA/RS simplices $\sum_{k}\alpha_k+\alpha_{\rm{c}}\le 1$ and $\sum_{k}\tau_k\le 1$, the order constraints $\alpha_{\rm{c}}\ge\alpha_k$, and per-user QoS constraints. When users are served in near/far pairs on orthogonal resources, each group is exactly \eqref{eq:Original_EC_Opt}, so the reformulation \eqref{eq:46} and the SCA surrogates \eqref{eq:48}--\eqref{eq:51} apply per group unchanged, the convexified problem \eqref{eq:EC_Opt2} stays convex, and the BCD/SCA convergence is preserved; the per-group decomposition keeps the cubic CVX cost manageable as $K$ grows. A full single-cell $K$-user extension, where one common stream is shared by all users and each private rate sees interference from the other $K-1$ streams, requires re-derived surrogate bounds and joint clustering/SIC-ordering optimization, and is left as future work (Section~\ref{Sect:FutureWork}).
\end{remark}

\section{Data-Driven Resource Allocation Framework}
\label{Sect:Data-Driven}

\subsection{ML-Based Resource Allocation Applications}
\label{Sect:ML_App}
ML has recently emerged as a powerful tool for enabling real-time decision-making in wireless communication systems, particularly in scenarios where model-based optimization methods incur prohibitive computational complexity. In STAR-RIS-assisted RSMA IoT systems, the fairness-oriented RA problems formulated in Section~\ref{Sect:Model_based_RA} require iterative convex approximations and BCD procedures, which result in non-negligible execution latency and hinder their real-time configurations.
To address this limitation, this work adopts a data-driven ML framework to approximate the optimal RA policy learned from model-based optimization.\footnote{We deliberately cast RA as a \emph{supervised} regression that imitates the optimization-based ground-truth labels, which yields millisecond inference without the thousands of training episodes and the training instability of DRL \cite{maghrebi2024deep,amiri2025resource}. DRL and graph neural networks (GNNs) address a different regime, namely label-free online policy learning and topology-varying graphs, and are thus not directly comparable to our label-imitation regressors; extending the framework with GNNs for topology-varying massive multi-user deployments is left as future work (Section~\ref{Sect:FutureWork}).}
Instead of repeatedly solving computationally intensive optimization problems online, ML models are trained offline to capture the nonlinear mapping between system parameters and the optimal RA configuration. Once trained, the ML models can instantaneously infer the RA solution with significantly reduced computational overhead, enabling real-time implementation.

As illustrated in Fig.~\ref{Fig:ML_App}, the three proposed frameworks operate in an interconnected pipeline. The information-theoretic analysis in Section~\ref{Sect:Analytical} produces the closed-form OP/EC and asymptotic metrics, which serve both as the objectives/constraints of the model-based optimization and as the post-prediction fairness evaluator. The model-based optimization in Section~\ref{Sect:Model_based_RA} consumes these analytical expressions to solve the OP/EC fairness problems via SCA and BCD, thereby generating the ground-truth RA labels. The data-driven learning in Section~\ref{Sect:Data-Driven} then uses these labels to train the DMNN/CMNN/MXGB models for real-time inference, whose predicted PA/RS coefficients are finally evaluated through the analytical OP/EC expressions, closing the loop among analysis, optimization, and learning.
Within this pipeline, the proposed ML framework operates in two distinct phases:

    \textbf{1) Offline processing}: In this phase, large-scale datasets are first generated using the model-based RA solutions developed in Section~\ref{Sect:Model_based_RA}. These labeled datasets are then used to train and optimize ML models capable of accurately approximating the optimal RA mapping.
    To enhance learning performance, a hyperparameter optimization (HPO) procedure is conducted prior to training. Specifically, a random search-based strategy \cite{bergstra2011algorithms} is adopted to explore the predefined hyperparameter space, where candidate hyperparameters are randomly initialized and iteratively refined through stepwise searches. The objective of the HPO process is to minimize the mean squared error (MSE) between the predicted and optimal RA outputs, and the procedure terminates once a predefined stopping criterion is satisfied. This offline training and optimization phase is computationally intensive but can be executed without real-time constraints on general-purpose computing platforms.

    \textbf{2) Online prediction}: In the online phase, the well-trained ML models are deployed at the network controller or base station to enable real-time RA decisions. Given the instantaneous system configuration, the ML inference module directly outputs near-optimal PA and RS coefficients with negligible computational latency. This significantly reduces the execution time compared to model-based optimization approaches and enables practical real-time implementation while maintaining fairness performance.

It is worth clarifying that the fairness performance is evaluated after predicting the PA and RS coefficients, rather than being directly predicted by the ML models. This is because the OP fairness values are extremely small in magnitude (e.g., on the order of $10^{-4}$ or smaller), and direct prediction would cause these values to be approximated as zeros, leading to unreliable OP fairness evaluation.

\subsection{Offline Processing}

\subsubsection{Dataset Generation}
To train the ML models for real-time RA under outage and capacity fairness optimization, two large-scale supervised datasets are constructed using the model-based fairness-oriented RA solutions developed in Section~\ref{Sect:Model_based_RA}.
Specifically, the optimal PA and RS coefficients obtained by solving \eqref{eq:Original_OP_Opt} via Algorithm~\ref{Algo:Overall_algorithm} in Section~\ref{Sect:OutageFairness} are treated as ground-truth labels for the first dataset.
Meanwhile, the labels for the second dataset are generated from the optimal PA and RS coefficients obtained by solving \eqref{eq:EC_Opt2} in Section~\ref{Sect:CapacityFairness} using CVX.

For each data sample, the input feature vector is randomly generated within predefined ranges that reflect practical deployment scenarios. Specifically,
$\bar \gamma \in [0,50]$~dB,
$N \in [5,100]$,
$\{ w_{\rm{n}},w_{\rm{f}} \} \in \{ 0,1 \}$,
$\{ \delta_i^{\rm{R}},\delta_i^{\rm{T}} \} \in [0,50]$~dB,
$\beta_i^{\rm{R}} \in[0.5,1]$,
and
$\{ R_{\rm{c}}^{\rm{th}},R_{\ell}^{\rm{th}} \} \in [0.3, 1.5]$~bps/Hz.
Each data sample is represented as a row vector pair $[{\bf{x}}_{[s]}^T,{\bf{y}}_{[s]}^T]$,
where
${\bf{x}}_s = [\bar\gamma, N, w_{\rm{n}},w_{\rm{f}}, \delta_i^{\rm{R}},\delta_i^{\rm{T}}, \beta_i^{\rm{R}},R_{\rm{c}}^{\rm{th}},R_{\rm{n}}^{\rm{th}},R_{\rm{f}}^{\rm{th}} ]^T$
and
${\bf{y}}_s = [{{\alpha _{\rm{c}}},{\alpha _{\rm{n}}},{\alpha _{\rm{f}}}},{{\tau _{\rm{n}}},{\tau _{\rm{f}}}}]^T$ denote the input feature and output label vectors of the data sample $s$, respectively.
The dataset contains 500000 samples, with 80\% used for training and the remaining 20\% reserved for testing.

Prior to training, both the input features and output labels are normalized to the interval $[0,1]$ to enhance numerical stability, accelerate convergence, and mitigate overfitting.

\subsubsection{Description of Machine Learning Models}

\begin{figure*}[!t]
\centering
\subfloat[DMNN model]{\fbox{\includegraphics[width=0.40\linewidth]{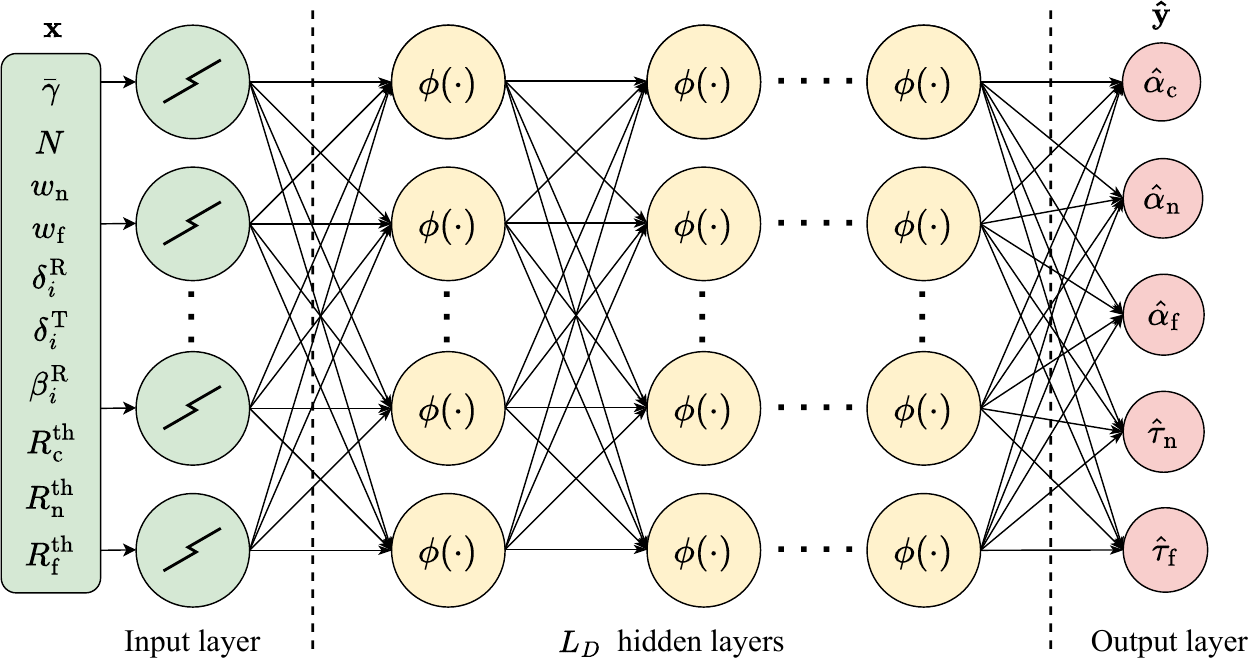}}}\hfill
\subfloat[CMNN model]{\fbox{\includegraphics[width=0.55\linewidth]{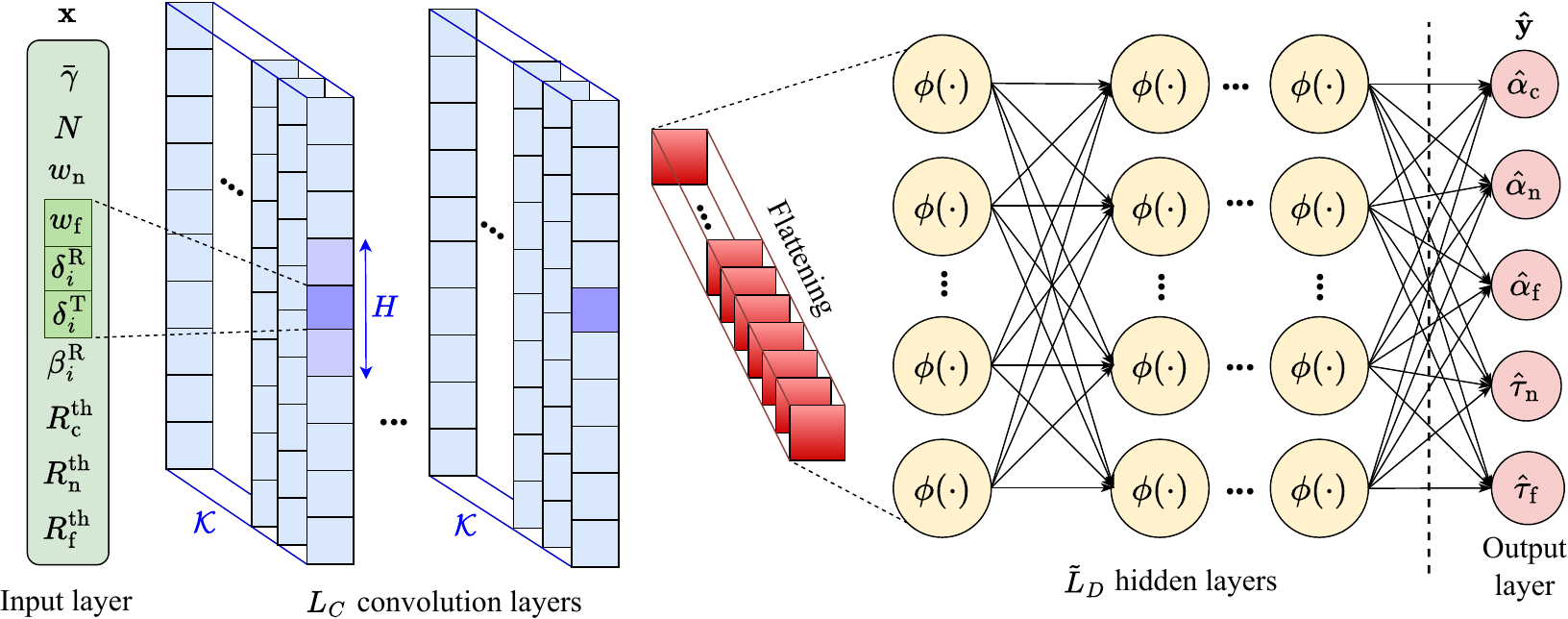}}}\\
\subfloat[MXGB model]{\fbox{\includegraphics[width=0.26\linewidth]{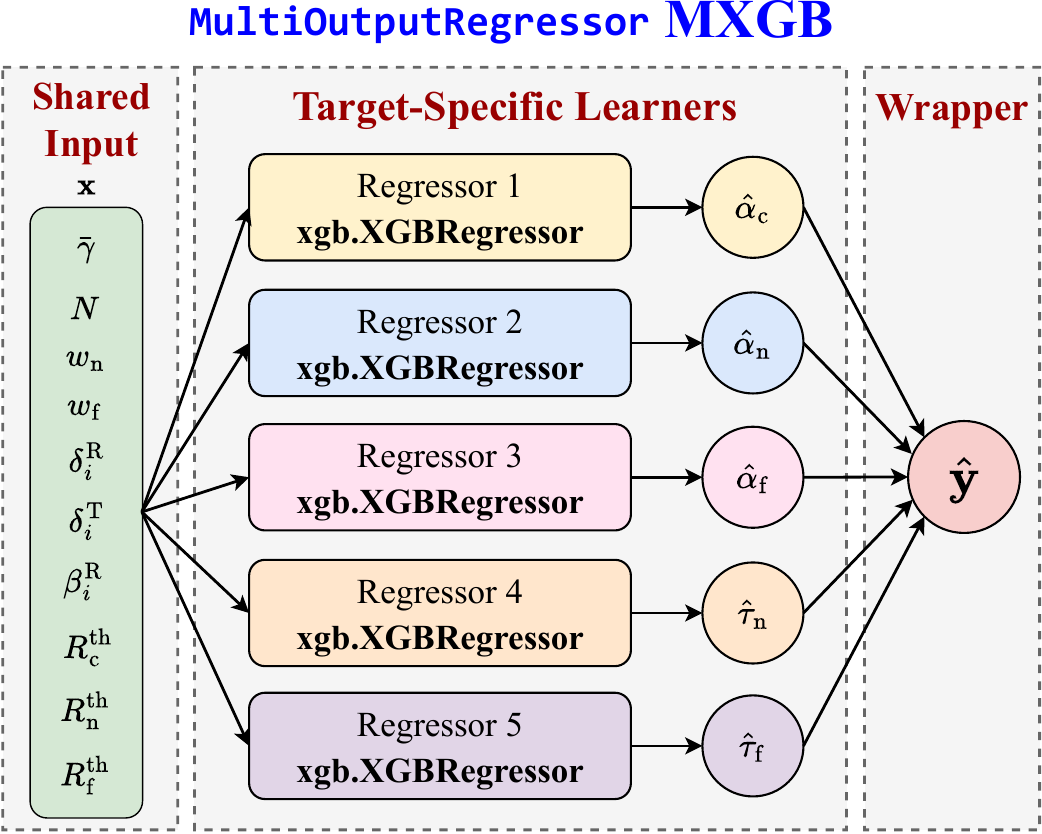}}}\hfill
\subfloat[Single tree-based regressor model]{\fbox{\includegraphics[width=0.70\linewidth]{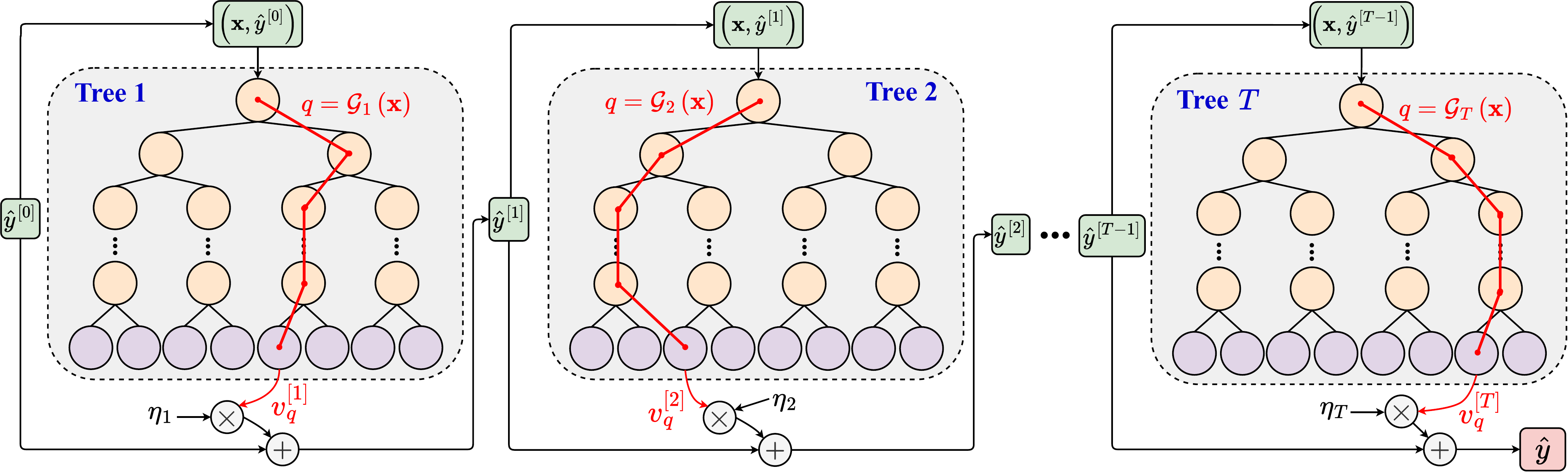}}}
\caption{Illustrations of (a) the DMNN, (b) the CMNN, (c) the MXGB, and (d) the single-target XGBoost model.}
\label{Fig:ML_Models}
\end{figure*}

\paragraph{DMNN}
The DMNN extends traditional deep neural networks by jointly learning multiple output variables within a single model.
As shown in Fig.~\ref{Fig:ML_Models}(a), the DMNN employed in this work is a feedforward architecture designed to learn the nonlinear mapping between system-level inputs and the optimal RA parameters.
The network consists of three main components: an input layer with $M_0 = 10$ neurons, $L_D$ fully connected layers (FCLs) with $M_l$ hidden neurons at layer $l$, and an output layer with $M_{L_D+1} = 5$ neurons.
Each hidden layer performs an affine transformation followed by a nonlinear activation. Specifically, the output of the
$l$-th hidden layer is given by
${{\bf{x}}^{\left[ l \right]}} = \phi \left( {{\bf{W}}_D^{\left[ l \right]}{{\bf{x}}^{\left[ {l - 1} \right]}} + {\bf{b}}_D^{\left[ l \right]}} \right),\quad \forall l \in \{ 1,...,{L_D+1} \}$,
where ${{\bf{W}}_D^{\left[ l \right]}} \in \mathbb{R}^{M_l \times M_{l-1}}$
and
${{\bf{b}}_D^{\left[ l \right]}} \in \mathbb{R}^{M_l \times 1}$
denote the weight matrix and bias vector of the $l$-th layer, respectively,
${{\bf{x}}^{\left[ 0 \right]}} \in \mathbb{R}^{M_l \times 1}$ is the initial input,
${\bf{\hat y}} = {{\bf{x}}^{\left[ {{L_D} + 1} \right]}}$ is the predicted output,
and
$\phi(\cdot)$ represents the scaled exponential linear unit (SELU) activation function \cite{clevert2015fast, klambauer2017self}.
The SELU activation is adopted due to its self-normalizing property, which improves training stability and mitigates vanishing-gradient issues \cite{klambauer2017self}.
Network parameters are optimized using the adaptive moment estimation (Adam) algorithm \cite{kingma2014adam} through backpropagation to minimize the mean squared error (MSE) between the predicted outputs and the ground-truth labels obtained from the model-based optimization framework.
The loss function of this supervised regression problem is defined as
${\cal L} = \frac{1}{{5S}}\sum\nolimits_{s = 1}^S {\left\| {{{{\bf{\hat y}}}_s} - {{\bf{y}}_s}} \right\|_2^2}$,
where $S$ denotes the number of training samples.

\paragraph{CMNN} \label{Sect:CMNN}
Similar to the DMNN, the CMNN extends conventional convolutional neural networks to the considered RA task.
As illustrated in Fig.~\ref{Fig:ML_Models}(b), the CMNN is designed to exploit local feature correlations among the system parameters while maintaining low inference complexity.
The CMNN architecture consists of a single input layer,
$L_C$ convolutional layers (CLs), one flattening layer,
$\tilde L_D$ FCLs, and a final output layer that jointly predicts the RA variables.
In the $l$-th CL, $\cal K$ convolutional filters with kernel size $H \times 1$ are applied to the input feature representation to extract structured characteristics of the system state. The output of the $k$-th filter channel at layer $l$ is given by
${{\bf{Z}}^{\left[ {l,k} \right]}} = \phi \left( {\sum\nolimits_{k' = 1}^{\cal K} {{\bf{W}}_C^{\left[ {l,k,k'} \right]} * {{\bf{X}}^{\left[ {l - 1,k'} \right]}}}  + {{b}}_C^{\left[ {l,k} \right]}} \right)$,
where $*$ denotes the convolution operator,
${{\bf{W}}_C^{\left[ {l,k} \right]}}$ and ${{{b}}_C^{\left[ {l,k} \right]}}$ represent the convolution kernel and the scalar bias associated with the $k$-th filter of CL $l$, respectively,
and ${{{\bf{X}}^{\left[ {l - 1} \right]}}}$
denotes the input feature map to the $l$-th layer, with ${{{\bf{X}}^{\left[ {0} \right]}}}$ being the initial input.
After completing feature extraction through all $L_C$ CLs, the resulting feature maps are reshaped into a one-dimensional vector via a flattening operation.
The flattened feature vector is subsequently processed by $\tilde L_D$ fully connected layers, which follow the same architectural design and activation strategy as those used in the DMNN model.

\begin{lemma}\label{Lemma:2}
    It is worth clarifying that the input feature vector is low-dimensional tabular data rather than image-like data. The CMNN is therefore included as a representative convolutional baseline to probe whether local correlations among the coupled system features can be exploited and to enable a fair architectural comparison against the DMNN and MXGB. To accommodate the small input dimensionality, single-stride convolutions with zero padding are adopted to preserve feature resolution and avoid information loss. As shown in Section~\ref{Sect:Results}, the tree-based MXGB outperforms both neural architectures in the accuracy--latency tradeoff, which is consistent with the known superiority of gradient-boosted trees on structured tabular data~\cite{shwartz2022tabular}.
\end{lemma}

\paragraph{MXGB}
The MXGB model is adopted as a tree-based ensemble learning approach for predicting the RA parameters. As illustrated in Fig.~\ref{Fig:ML_Models}(c), MXGB decomposes the multi-output regression task into multiple single-target boosting models that are trained independently and combined through a multi-output wrapper\footnote{It is worth noting that, at the time of writing, the state-of-the-art MXGB model is still under development, and support for training with vector-valued leaves for multi-output regression has not yet been implemented \cite{XGB_multioutput_regression}. Consequently, we adopt a one-model-per-target XGBoost approach using \texttt{MultiOutputRegressor}, which combines multiple single-output models within a multi-output wrapper.}, i.e., \texttt{MultiOutputRegressor} in the scikit-learn library \cite{MultiOutputRegressor_sklearn}. This design enables efficient learning while preserving strong prediction accuracy, particularly for tabular data.
As shown in Fig.~\ref{Fig:ML_Models}(d), each regressor consists of $T$ decision trees constructed sequentially. For simplicity, all trees are configured with identical depth $\cal D$, resulting in $L_X = 2^{\cal D}$ terminal leaves per tree.
For the $t$-th tree, each leaf node stores a real-valued prediction weight denoted by $v_q^{[t]}$, where $q \in \{ 1,...,L_X \}$.
Given an input feature vector $\bf{x}$, a deterministic decision rule ${\cal G}_t({\bf{x}})$ maps the input to a specific leaf index $q$ in the $t$-th tree. The output of the ensemble after the $t$-th boosting iteration is expressed recursively as
${\hat y^{\left[ t \right]}} = {\hat y^{\left[ {t - 1} \right]}} + {\eta _t}v_q^{\left[ t \right]}, y \in \{ {{\alpha _{\rm{c}}},{\alpha _{\rm{n}}},{\alpha _{\rm{f}}}},{{\tau _{\rm{n}}},{\tau _{\rm{f}}}} \}$,
where ${{\hat y}^{\left[ 0 \right]}}$ is initialized as the empirical mean of the training labels and
$\eta_t$ is a regularization coefficient controlling the contribution of the $t$-th tree.
The regularized MSE loss function at the $t$-th tree is given by
${{\cal L}_t} = \sum\nolimits_{q = 1}^{{L_X}} \Bigg\{ {\sum\nolimits_{s \in {\cal S}_q^{\left[ t \right]}} {\frac{2}{S}\left( {\hat y_s^{\left[ {t - 1} \right]} - {y_s}} \right)v_q^{\left[ t \right]}}} + \left( {\sum\nolimits_{s \in {\cal S}_q^{\left[ t \right]}} {\frac{1}{S}}  + \frac{\eta _t}{2}} \right){{\left( {v_q^{\left[ t \right]}} \right)}^2} \Bigg\}$,
where ${{\cal S}_q^{\left[ t \right]}}$ denotes the set of training samples routed to leaf $q$ in tree $t$.
After completing all $T$ boosting iterations, the final MXGB prediction for each output component is given by
$\hat y = {{\hat y}^{\left[ T \right]}} = {{\hat y}^{\left[ 0 \right]}} + \sum\nolimits_{t = 1}^T {{\eta _t}v_q^{\left[ t \right]}}$.

\subsection{Real-Time Inference Phase}
\label{Sect:RealTime}

After completing the offline training stage, the trained ML models are deployed for real-time RA inference. Given a new system configuration represented by the input feature vector $\mathbf{x}$, the RA parameters are directly obtained through a learned mapping
$\hat{\bf{y}} = \mathcal{F}(\mathbf{x})$,
where $\mathcal{F}(\cdot)$ represents the inference function implemented by the trained DMNN, CMNN, or MXGB model.

By replacing iterative optimization with a single forward inference, the proposed ML-based framework significantly reduces computational complexity and execution latency. Owing to its low inference cost, the real-time prediction phase can be efficiently implemented at network nodes, enabling fast and adaptive RA decisions across a variety of system configurations.
We clarify that the reported millisecond-scale figures reflect the algorithmic inference latency measured in simulation and thus indicate real-time capability rather than a fully validated real-time deployment. They do not yet account for the STAR-RIS control and reconfiguration latency, the fronthaul/control signaling, or the RF front-end effects. Establishing end-to-end real-time operation would require over-the-air, software-defined-radio, or hardware-emulation validation together with a realistic characterization of the STAR-RIS control latency, which we identify as important future work (Section~\ref{Sect:FutureWork}).

\subsection{Computational Complexity Analysis}\label{Sect:IVD}

The computational complexity of the proposed ML-based RA framework is measured in terms of the number of floating-point operations (FLOPs) required during the inference stage. Since real-time operation is dominated by forward inference rather than offline training, we focus on the inference complexity of each ML model.

For the DMNN, the inference complexity is dominated by the FCLs.
Each connection between layer $l-1$ and layer $l$ incurs two FLOPs: one multiplication and one addition accumulated into the neuron pre-activation.
Consequently, the total number of FLOPs required for a forward pass is given by
${\mathcal{O}}_{D} = 2 \sum\nolimits_{l=1}^{L_D+1} M_{l-1} M_l$.

For the CMNN, the inference complexity consists of convolutional, flattening, and fully connected operations.
The first CL requires $2M_0 H {\cal K}$ FLOPs, while each subsequent CL incurs $2M_0 H {\cal K}^2$ FLOPs.
After feature extraction through $L_C$ CLs, the flattening operation requires $M_0 {\cal K}$ FLOPs. The FCLs then contribute $2 \sum_{l=2}^{\tilde{L}_D+1} \tilde{M}_{l-1}\tilde{M}_l$ FLOPs, where
$\tilde M_l$ denotes the number of hidden neurons in FCL $l$,
$\tilde M_1 = M_0 {\cal K}$,
and $\tilde M_{\tilde L_D+1} = 5$.
Consequently, the overall complexity of the CMNN is expressed as
${{\cal O}_{\mathrm{CMNN}}} = 2{M_0}H{\cal K} + \left( {{L_C} - 1} \right)2{M_0}H{{\cal K}^2} + {M_0}{\cal K}
+ 2\sum\nolimits_{l = 2}^{{{\tilde L}_D} + 1} {{{\tilde M}_{l - 1}}{{\tilde M}_l}}$.

For the MXGB, the inference complexity is dominated by node comparisons in the decision trees.
Since the MXGB predicts five RA parameters independently, the total inference complexity is computed as
$\mathcal{O}_{\mathrm{MXGB}} = 5\left( {\sum\nolimits_{l = 1}^{{T}} {\left( {{\cal D} - 1} \right)}  + 2{T}} \right) =5 T ({\cal D} + 1)$.

While the inference FLOPs above are the metric relevant to real-time deployment, we additionally outline the offline-pipeline cost for completeness. \emph{Dataset generation} is the dominant offline cost: producing the $S=500{,}000$ labeled samples requires solving the model-based problems, i.e., ${\cal O}\big(S\,I_0({\cal C}_1+{\cal C}_2)\big)$ operations for the outage-fairness dataset via Algorithm~\ref{Algo:Overall_algorithm} and ${\cal O}(S\,{\cal C}_3)$ for the capacity-fairness dataset; this is a one-time, fully parallelizable expenditure. \emph{Training} costs ${\cal O}(N_{\rm ep}\,S\,{\cal O}_{\rm model})$ for the neural models, where $N_{\rm ep}$ is the number of epochs and a forward-backward pass is about three times the inference FLOPs per sample, and ${\cal O}(T\,S\log S)$ per output for the gradient-boosted MXGB. \emph{Memory} scales with the parameter/leaf counts, i.e., $\sum_l M_{l-1}M_l$ (DMNN), the convolutional and FCL parameters (CMNN), and $5T2^{\cal D}$ leaf values (MXGB). The associated \emph{energy overhead} is a one-time offline GPU/CPU expenditure. Crucially, all of these costs are incurred once offline and are amortized over the subsequent millisecond-scale inferences, so they do not affect the real-time operation, which is governed solely by ${\cal O}_D$, ${\cal O}_{\rm CMNN}$, and ${\cal O}_{\rm MXGB}$ above.

\section{Numerical Results and Discussion}
\label{Sect:Results}

This section presents Monte Carlo simulations to validate the proposed information-theoretic analysis and RA frameworks, from which system-specific insights are discussed.
Unless otherwise stated, the simulation parameters are set as follows.
The network nodes are placed on a two-dimensional plane with coordinates
${\cal S}(0.05,0.1)$,
${\cal U}_{\rm{n}}(0,0)$,
${\cal U}_{\rm{f}}(0.25,0)$,
and
${\cal R}(0.1,0.04)$,
where the plane is normalized with a scaling factor of $r=100$~m.
For example, the physical distance between $\cal S$ and $\cal R$ is calculated as
${\rm{d}}_h = r \times \sqrt{(x_{\cal S} - x_{\cal R})^2 + (y_{\cal S} - y_{\cal R})^2}$ \cite{tu2024irs},
where $(x_{\cal Y},y_{\cal Y})$ denotes the coordinates of node ${\cal Y} \in \{ {\cal S}, {\cal R}, {\cal U}_{\rm{n}}, {\cal U}_{\rm{f}} \}$.
The distances ${\rm{d}}_g$, ${\rm{d}}_t$, and ${\rm{d}}_\ell$ are computed using the same approach.
The path-loss exponent factor is set to $\eta = 2.5$ \cite[Table~2.2]{goldsmith2005wireless}.
Additionally, we set
$w_{\rm{n}} = w_{\rm{f}} = 1$,
$N = 30$,
$\beta_i^{\rm{R}}=0.6$,
$\beta_i^{\rm{T}}=0.4$,
$\delta_i^{\rm{R}}= \delta_i^{\rm{T}}=10$,
$\alpha_{\rm{c}} = 0.7$,
$\alpha_{\rm{n}} = 0.6(1-\alpha_{\rm{c}})$,
$\alpha_{\rm{f}} = 0.4(1-\alpha_{\rm{c}})$,
$\tau_{\rm{n}} = 0.7$,
$\tau_{\rm{f}} = 0.3$,
$R_{\rm{c}}^{{\rm{th}}} = 0.5$~bps/Hz,
$R_{\rm{n}}^{{\rm{th}}} = 0.4$~bps/Hz,
and
$R_{\rm{f}}^{{\rm{th}}} = 0.6$~bps/Hz \cite{tu2025active}.
The hyperparameter settings of the DMNN, CMNN, and MXGB models are summarized in Table~\ref{Tab:Hyperparameters}.

\begin{table}[!t]
\centering
\caption{Tuning and training hyperparameters of the DMNN, CMNN, and MXGB models.}
\label{Tab:Hyperparameters}
\footnotesize
\renewcommand{\arraystretch}{1.25}
\setlength{\tabcolsep}{4pt}
\begin{tabular}{|l|l|p{2.4cm}|}
\hline
\textbf{Hyperparameter} & \textbf{Model} & \textbf{Setting / Domain} \\
\hline\hline
\multicolumn{3}{|l|}{\textit{Architecture (tuned via HPO; optimal values in Table~\ref{Tab:HPO})}} \\
\hline
Number of FCLs $\{L_D,\tilde L_D\}$ & DMNN, CMNN & $[1,6]$ \\ \hline
Number of CLs $L_C$ & CMNN & $[1,6]$ \\ \hline
Hidden neurons $\{M_l,\tilde M_l\}$ & DMNN, CMNN & $[25,200]$ \\ \hline
Number of filters ${\cal K}$ & CMNN & $[4,64]$ \\ \hline
Kernel size $H$ & CMNN & $[2,32]$ \\ \hline
Number of trees $T$ / depth ${\cal D}$ & MXGB & $[1,100]$ \\ \hline
\multicolumn{3}{|l|}{\textit{Training configuration}} \\
\hline
Activation & DMNN, CMNN & SELU \\ \hline
Optimizer & DMNN, CMNN & Adam \\ \hline
Initial learning rate & DMNN, CMNN & $10^{-3}$ \\ \hline
LR decay / scheduler & DMNN, CMNN & $\times 0.9$, patience $5$ \\ \hline
Minimum learning rate & DMNN, CMNN & $10^{-4}$ \\ \hline
Early-stopping patience & DMNN, CMNN & $10$ epochs \\ \hline
Learning rate & MXGB & $0.3$ (fixed) \\ \hline
Early stopping & MXGB & Enabled \\ \hline
Loss function & All & MSE \\ \hline
HPO method & All & Random search \\ \hline
Convolution stride / padding & CMNN & $1$ / zero \\ \hline
\multicolumn{3}{|l|}{\textit{Dataset}} \\
\hline
Number of samples & All & $500{,}000$ \\ \hline
Train/test split & All & $80\%/20\%$ \\ \hline
Feature/label normalization & All & $[0,1]$ (min--max) \\ \hline
\end{tabular}
\end{table}

\subsection{Analytical Performance Evaluation}

\subsubsection{Graphical Validation of Assumption~\ref{Assumption:1} and Active-Passive STAR-RIS Comparison}\label{Sect:5.1}

Fig.~\ref{Fig:OP_EC_Assumption1_Validation} illustrates the OP and EC performance versus the average SNR $\bar \gamma$, providing a graphical validation of the STAR-RIS SI mitigation and comparing the active-passive STAR-RIS modes.
Details are discussed below.

\begin{figure}[!t]
\centering
\subfloat[OP performance]{\includegraphics[width=\linewidth]{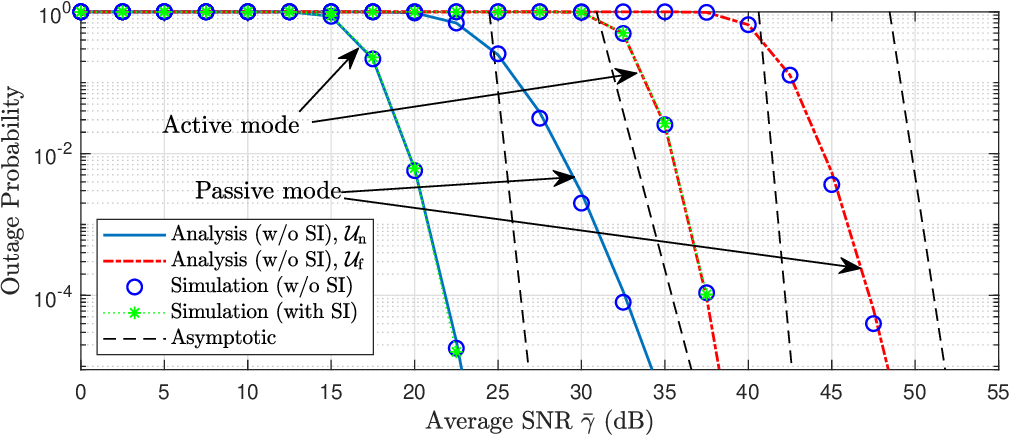}}\\
\subfloat[EC performance]{\includegraphics[width=\linewidth]{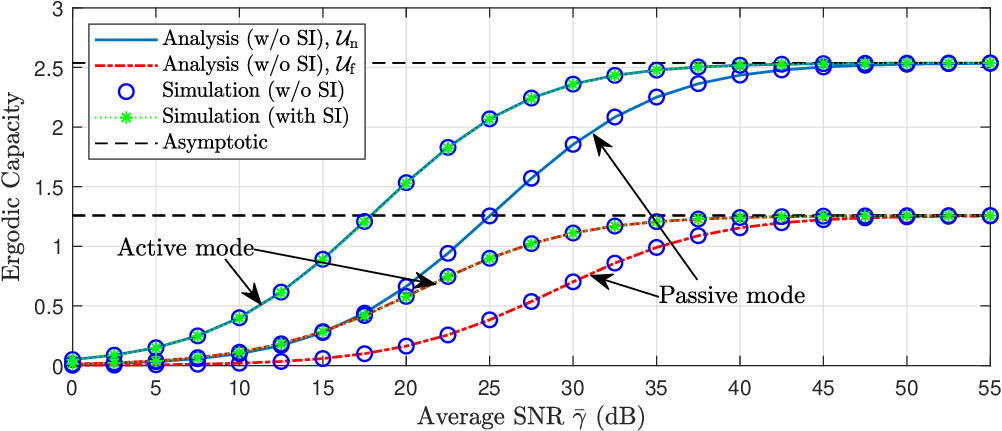}}
\caption{OP and EC versus average SNR: Graphical validation of Assumption~\ref{Assumption:1} and active-passive STAR-RIS comparison.}
\label{Fig:OP_EC_Assumption1_Validation}
\end{figure}

\paragraph{Graphical Validation of Assumption~\ref{Assumption:1}}
When the STAR-RIS SI signal is included in \eqref{eq:R_cl} and \eqref{eq:R_l}, the achievable rates for decoding ${s_{\rm{c}}}$ and $s_\ell$ at $\cal U_\ell$ become, respectively,
\begin{align}\label{eq:R_cl_SI}
    {\Re _{{\rm{c}},\ell }^{\rm{SI}}} &= {\tau _\ell }{\log _2}\left( {1 + \frac{{{\kappa _\ell }{\alpha _{\rm{c}}}\bar \gamma }}{{{\kappa _\ell }(1 - {\alpha _{\rm{c}}})\bar \gamma + Y_\ell + 1}}} \right),\\
    \label{eq:R_l_SI}
    {\Re _\ell^{\rm{SI}} } &= {\log _2}\left( {1 + \frac{{{\kappa _\ell }{\alpha _\ell }\bar \gamma }}{{{\kappa _\ell }{\alpha _{\ell '}}\bar \gamma  + Y_\ell + 1}}} \right), \forall \ell' \ne \ell,
\end{align}
where ${Y_{\rm{n}}} = {\left| {{{\bf{g}}^T}{{\bf{\Phi }}_R}{{\bf{n}}_{{\rm{SI}}}}} \right|^2}/{\sigma ^2}$,
${Y_{\rm{f}}} = {\left| {{{\bf{t}}^T}{{\bf{\Phi }}_T}{{\bf{n}}_{{\rm{SI}}}}} \right|^2}/{\sigma ^2}$,
and
${{\bf{n}}_{{\rm{SI}}}} \sim {\cal C}{\cal N}\left( {{\bf{0}},\sigma _{{\rm{SI}}}^2}{\bf{I}}_N \right)$.
The noise power is modeled as $\sigma^2 = -174 + 10\log_{10}({\rm{BW}}) + {\rm{NF}}$, where the noise figure is ${\rm{NF}} = 10$~dB and the system bandwidth is ${\rm{BW}} = 10$ MHz \cite{do2021multi}.
The STAR-RIS SI power is set to ${\sigma _{{\rm{SI}}}^2} = -80$ dBm \cite{vu2025aerial}.
It is observed in Fig.~\ref{Fig:OP_EC_Assumption1_Validation} that, in the active mode, the OP and EC curves corresponding to \eqref{eq:R_cl_SI} and \eqref{eq:R_l_SI}, marked as “Simulation (with SI),” closely coincide with those obtained from \eqref{eq:R_cl} and \eqref{eq:R_l} over the entire SNR range.

To complement this graphical validation with a quantitative argument, we bound the rate penalty of neglecting the STAR-RIS SI. Comparing the SI-aware rates \eqref{eq:R_cl_SI}--\eqref{eq:R_l_SI} with \eqref{eq:R_cl}--\eqref{eq:R_l}, the SI-induced rate loss of the common stream at $\mathcal{U}_\ell$ is
\begin{align*}
 \tau_\ell\log_2\!\left(1+
\frac{\kappa_\ell\alpha_c\bar{\gamma}\,Y_\ell}
{I_{0,\ell}(I_{0,\ell}+Y_\ell+\kappa_\ell\alpha_c\bar{\gamma})}\right)
\leq \log_2\!\left(1+\frac{Y_\ell}{I_{0,\ell}}\right),
\end{align*}
where $I_{0,\ell}=\kappa_\ell(1-\alpha_c)\bar{\gamma}+1$ and the inequality uses $\tau_\ell\leq 1$ together with
$I_{0,\ell}+Y_\ell+\kappa_\ell\alpha_c\bar{\gamma}\geq\kappa_\ell\alpha_c\bar{\gamma}$; a similar bound with $I_{0,\ell}$ replaced by $\kappa_\ell\alpha_{\ell'}\bar{\gamma}+1$ holds for the private stream.
Since both $Y_\ell$ and the cascaded term in $I_{0,\ell}$ are amplified by the same STAR-RIS gain, the ratio $Y_\ell/I_{0,\ell}$ is independent of the amplification gain $\delta$ and the element count $N$, reducing to $Y_\ell/(\kappa_\ell(1-\alpha_c)\bar{\gamma})\sim
\sigma_{\mathrm{SI}}^2/\big((1-\alpha_c)P_S\big)$.
With $\sigma_{\mathrm{SI}}^2=-80$ dBm and $P_S\geq 30$ dBm, this ratio is on the order of $10^{-11}$, so the SI-induced rate loss is upper bounded on the order of $10^{-11}$ bps/Hz and is thus negligible.
This analytical bound corroborates the curve overlap observed in Fig.~\ref{Fig:OP_EC_Assumption1_Validation} and quantitatively validates Assumption~\ref{Assumption:1} within the practical gain
range.
The complementary high-gain regime, in which the residual SI is modeled explicitly and amplifier nonlinearity emerges as the dominant impairment, is numerically analyzed in Section~\ref{Sect:Residual_SI}.

\paragraph{Active-Passive STAR-RIS Comparison}
From a system-level perspective,
Fig.~\ref{Fig:OP_EC_Assumption1_Validation} highlights the performance gap between active and passive STAR-RIS modes. Switching from passive to active mode, together with increasing $\bar \gamma$, considerably improves the OP and EC performance of both users.
As shown in Fig.~\ref{Fig:OP_EC_Assumption1_Validation}(a), increasing $\bar \gamma$ continues to reduce OP values even in the high-SNR region, which is consistent with the discussion in Remark~\ref{remark:OP}.
In contrast, Fig.~\ref{Fig:OP_EC_Assumption1_Validation}(b) shows that the EC saturates at high SNR and converges to the asymptotic EC in \eqref{eq:asymptotic_EC}, regardless of whether the STAR-RIS operates in active or passive mode.
This saturation occurs because, at high SNR, Gaussian noise becomes negligible while both the desired signal and the inter-message interference scale proportionally with the transmit power, resulting in a constant signal-to-interference-plus-noise ratio.
This observation aligns with the analysis presented in Remark~\ref{remark:EC}.
Moreover, the analytical and simulation results exhibit excellent agreement, thereby validating the accuracy of the proposed information-theoretic analysis framework.

\subsubsection{Impact of Direct Links}

\begin{figure}[!t]
\centering
\subfloat[OP performance]{\includegraphics[width=\linewidth]{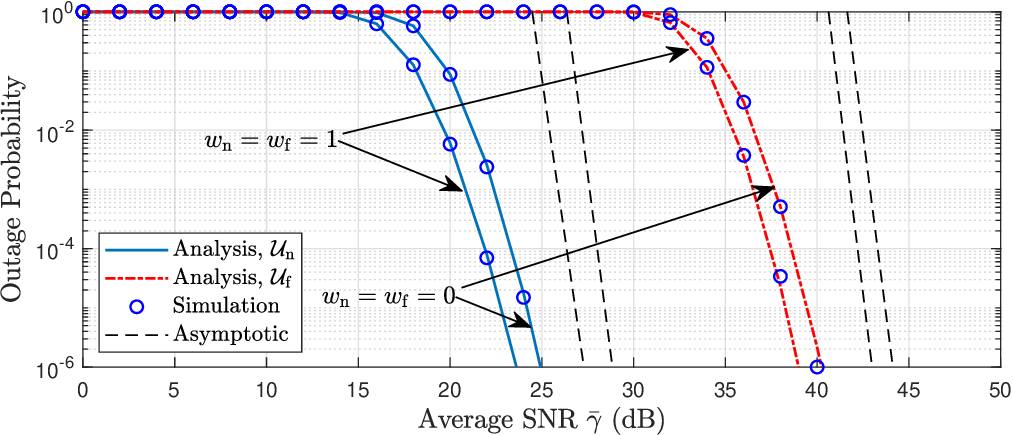}}\\
\subfloat[EC performance]{\includegraphics[width=\linewidth]{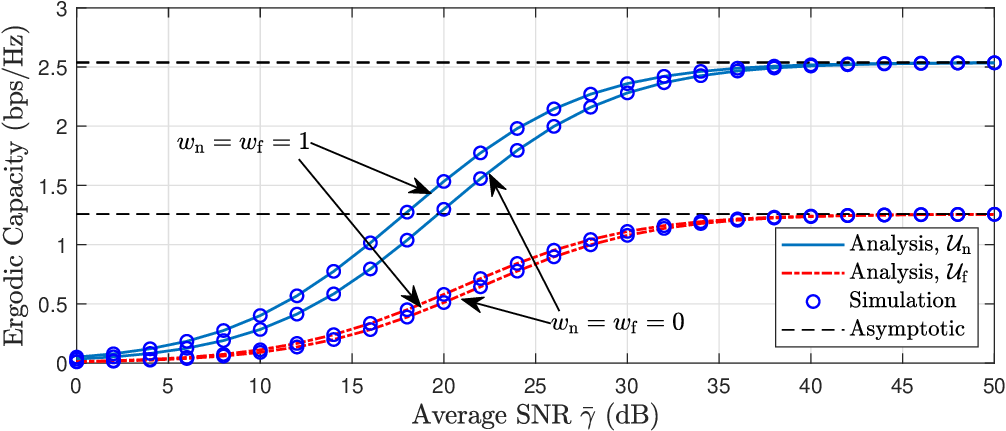}}
\caption{OP and EC versus average SNR: Comparison between the presence and absence of direct links.}
\label{Fig:OP_EC_SNR_Direct_NonDirect}
\end{figure}

Fig.~\ref{Fig:OP_EC_SNR_Direct_NonDirect} examines the impact of direct links between $\cal S$ and $\cal U_\ell$.
In addition to the cascaded STAR-RIS-assisted links, the presence of direct links
($w_{\rm{n}} = w_{\rm{f}} = 1$) enables cooperative communication, which enhances both OP and EC performance. However, this improvement is smaller than the performance gain attributed to the operating mode observed in Fig.~\ref{Fig:OP_EC_Assumption1_Validation}.
It is important to note that direct links between $\cal S$ and $\cal U_\ell$ are not always available in practice due to blockages or severe shadowing, whereas STAR-RIS-assisted transmission inherently suffers from the double-fading effect.
Nevertheless, the results in Figures~\ref{Fig:OP_EC_Assumption1_Validation} and \ref{Fig:OP_EC_SNR_Direct_NonDirect} indicate that the QoS requirements of the users can still be satisfied when an active STAR-RIS is deployed, where the amplification gain and the increased number of STAR-RIS elements effectively compensate for the double-fading cascaded channel.
Furthermore, regarding the asymptotic values in Fig.~\ref{Fig:OP_EC_SNR_Direct_NonDirect}, the two users behave differently: the OP of both keeps decaying with the nonzero diversity order $D_\ell=\xi_\ell/2$ (Proposition~\ref{proposition:diversity_order}), but the near user attains a larger $\xi_{\rm{n}}$ from its stronger channels and hence a steeper slope and lower OP; meanwhile, the EC of both saturates to the SNR-independent floor in \eqref{eq:asymptotic_EC}, whose user-dependent gap is set solely by the asymmetric PA/RS allocation rather than the transmit power.

\begin{figure}[!t]
\centering
\subfloat[OP performance]{\includegraphics[width=\linewidth]{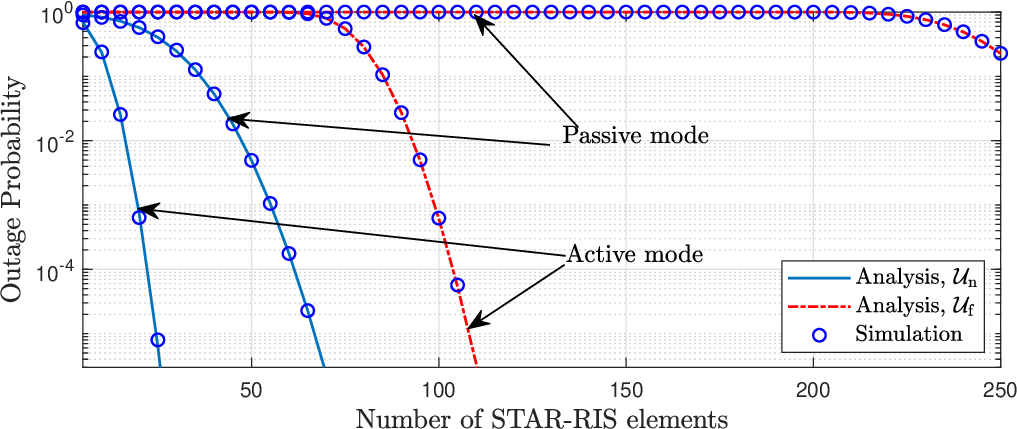}}\\
\subfloat[EC performance]{\includegraphics[width=\linewidth]{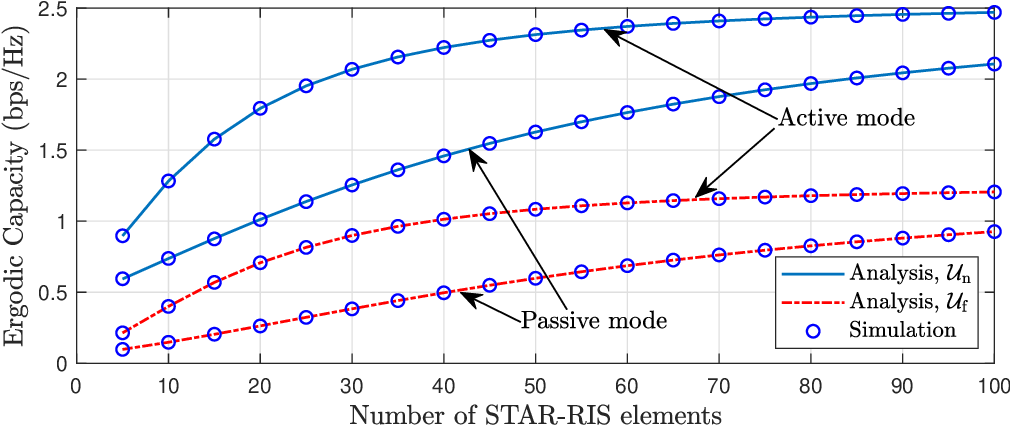}}
\caption{Impact of the number of STAR-RIS elements on the OP and EC, where $\bar \gamma = 25$ dB.}
\label{Fig:OP_EC_RIS_elements}
\end{figure}

\subsubsection{Impact of STAR-RIS Elements and Amplification Gain}

\begin{figure}[!t]
\centering
\subfloat[OP performance]{\includegraphics[width=\linewidth]{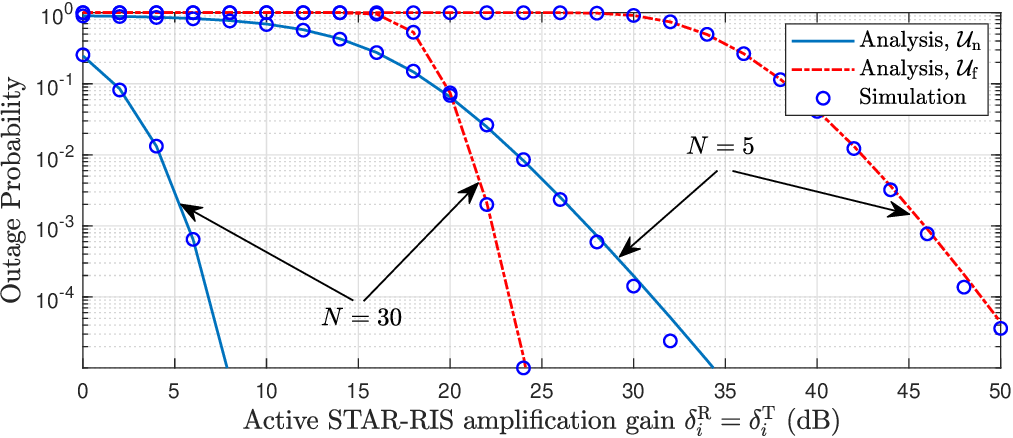}}\\
\subfloat[EC performance]{\includegraphics[width=\linewidth]{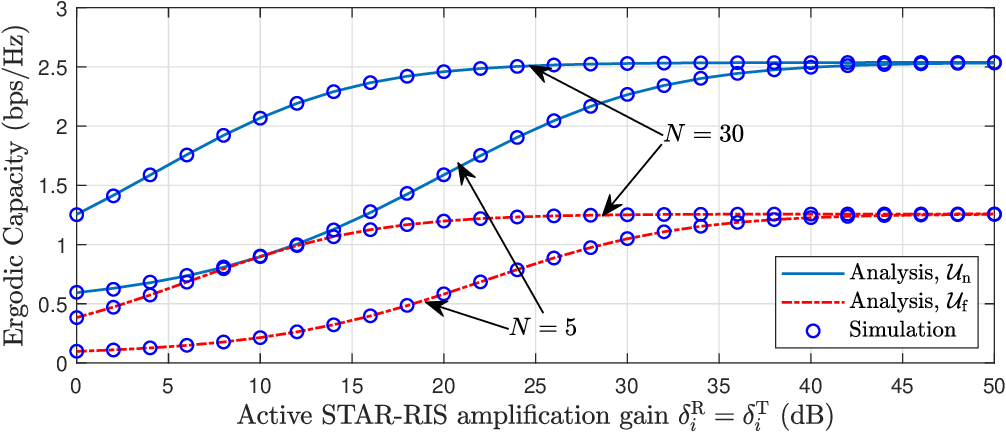}}
\caption{Impact of the active STAR-RIS amplification gain on the OP and EC, where $\bar \gamma = 25$ dB.}
\label{Fig:OP_EC_Amplification}
\end{figure}

Figures~\ref{Fig:OP_EC_RIS_elements} and \ref{Fig:OP_EC_Amplification} investigate the impact of the number of STAR-RIS elements and the STAR-RIS amplification gain, respectively, assuming $\bar \gamma = 25$ dB.
The results indicate that both OP and EC performance improve significantly as these parameters increase, which is consistent with the discussion in Remark~\ref{remark:OP}.
Moreover, the OP continues to decrease with larger parameter values, whereas the EC tends to saturate once the parameters reach sufficiently high levels.
In addition, the analytical curves closely match the simulation results across the entire parameter range, confirming the accuracy and robustness of the proposed information-theoretic analysis framework.

\subsubsection{On Performance of Energy Efficiencies}

\begin{figure}[!t]
\centering
\subfloat[Throughput-based EE]{\includegraphics[width=0.49\linewidth]{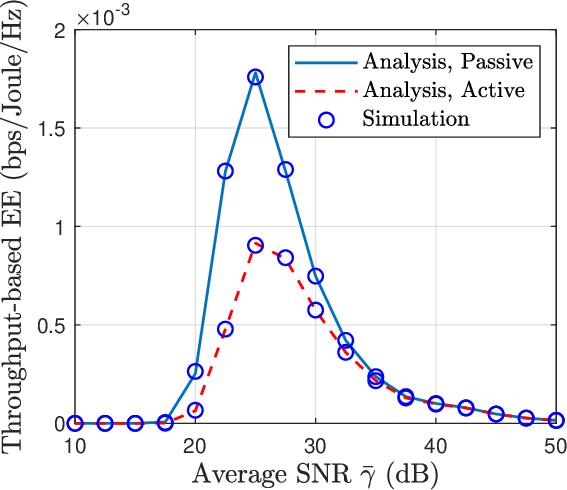}}\hfill
\subfloat[Spectral-based EE]{\includegraphics[width=0.49\linewidth]{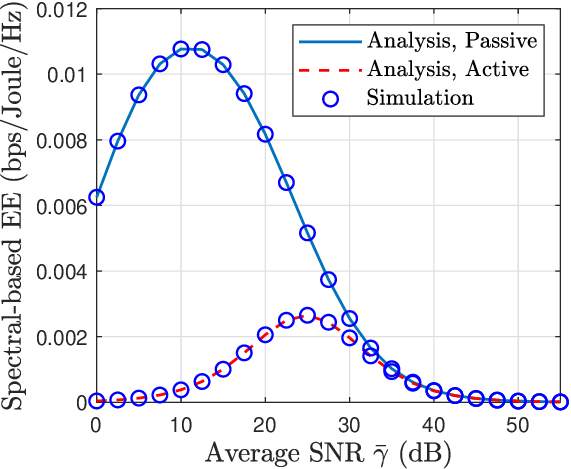}}
\caption{Performance tradeoff in terms of (a) throughput-based EE and (b) spectral-based EE.}
\label{Fig:EE}
\end{figure}

Fig.~\ref{Fig:EE} illustrates the performance tradeoff in terms of throughput-based and spectral-based EE, where the static circuit power $P_0$ is set to 1~W \cite{tu2024irs}.
It is observed that both EE metrics initially increase with the average SNR $\bar \gamma$, since higher SNR improves the probability of successful decoding and the EC, thereby enhancing the effective throughput and spectral efficiency, respectively. However, as $\bar \gamma$ continues to rise, the static and STAR-RIS-related power consumption dominates the denominator of the EE expressions, while the throughput and spectral efficiency offer negligible improvement and saturation, respectively.
Consequently, the EE expressions eventually decrease, exhibiting a unimodal behavior.
Furthermore, although the active mode provides superior OP and EC performance, the passive STAR-RIS achieves higher peak throughput-based and spectral-based EE because it requires no additional amplification power. This makes the passive configuration more energy-efficient in low-to-moderate SNR regimes, highlighting a fundamental tradeoff between capacity enhancement and energy efficiency.

\begin{figure}[!t]
\centering
\subfloat[OP performance]{\includegraphics[width=\linewidth]{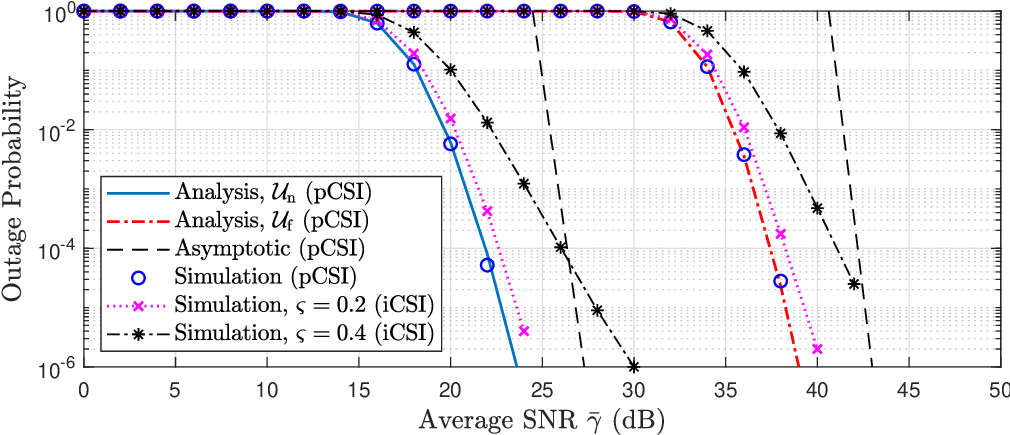}}\\
\subfloat[EC performance]{\includegraphics[width=\linewidth]{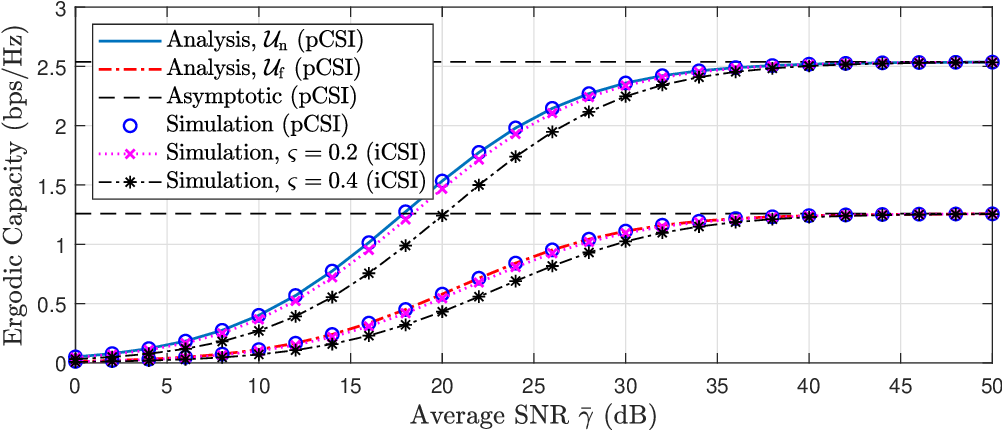}}
\caption{Impact of iCSI on OP and EC performance.}
\label{Fig:OP_EC_SNR_iCSI}
\end{figure}

\subsubsection{Impact of Imperfect CSI} \label{Sect:iCSI}
Although perfect CSI (pCSI) is assumed in the analytical derivations to establish a theoretical performance benchmark, practical STAR-RIS IoT systems can suffer from channel-estimation errors due to pilot contamination, limited feedback, and the double-fading cascaded structure. To capture this non-ideality, we numerically model the imperfect CSI (iCSI) as $\hat{x} = \sqrt{1-\varsigma^2}\,x + \varsigma\, e$, $\forall x\in\{h_i,g_i,t_i,d_\ell\}$, where $e\sim\mathcal{CN}(0,\lambda_x)$ is the zero-mean estimation error uncorrelated with $x$, and $\varsigma\in[0,1]$ denotes the CSI imperfection level. This model preserves the per-link average power and reduces to the perfect-CSI case when $\varsigma=0$.\footnote{It is worth mentioning that incorporating the iCSI directly into the theoretical derivations would require re-deriving the OP/EC expressions and would substantially alter the entire analytical and resource allocation frameworks; it is therefore retained as a numerical study only.}
Fig.~\ref{Fig:OP_EC_SNR_iCSI} simulates the OP and EC under $\varsigma\in\{0,0.2,0.4\}$.
The results indicate that for moderate CSI errors typical of standard estimation methods ($\varsigma\le0.2$), the performance remains close to the pCSI case, demonstrating the robustness of the proposed system, whereas larger errors lead to a graceful and predictable degradation.

\begin{figure}[!t]
\centering
\includegraphics[width=\linewidth]{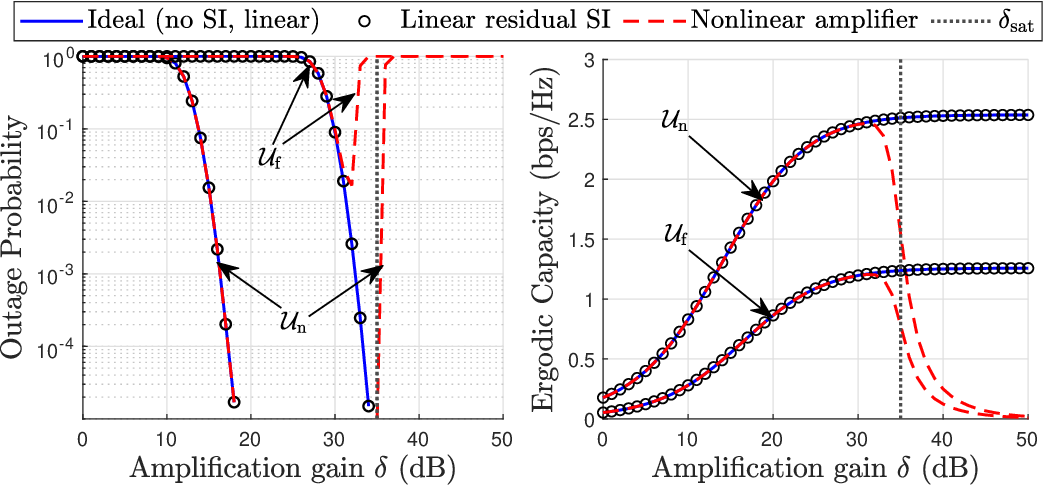}
\caption{Impact of residual SI and nonlinear amplification on OP and EC, where $\varrho = N = 30$, $p=3$, and $\delta_{\mathrm{sat}}=35$ dB.}
\label{Fig:non_linear}
\end{figure}

\subsubsection{Impact of Residual SI and Nonlinear Amplification}\label{Sect:Residual_SI}
Assumption~\ref{Assumption:1} is established within the practical operating region $\delta_i^{\rm{R}} = \delta_i^{\rm{T}} = \delta\le\delta_{\mathrm{sat}}$. To examine the high-gain regime, we first make the residual SI explicit. From \eqref{eq:R_cl_SI} and \eqref{eq:R_l_SI}, the residual SI at $\mathcal{U}_\ell$ is $Y_{\mathrm n}=|\mathbf{g}^{T}\boldsymbol{\Phi}_{\mathrm R}\mathbf{n}_{\mathrm{SI}}|^2/\sigma^2$
and
$Y_{\mathrm f}=|\mathbf{t}^{T}\boldsymbol{\Phi}_{\mathrm T}\mathbf{n}_{\mathrm{SI}}|^2/\sigma^2$,
with $\mathbf{n}_{\mathrm{SI}}\sim\mathcal{CN}(\mathbf{0},\sigma_{\mathrm{SI}}^2\mathbf{I}_N)$
\cite{vu2025aerial,vu2024hybrid}.
Since the SI components combine incoherently, the average residual SIs are
$\mathbb{E}\{Y_{\mathrm n}\}=(\sigma_{\mathrm{SI}}^2/\sigma^2)\lambda_g\beta_{\Sigma2}^{\mathrm R}$
and
$\mathbb{E}\{Y_{\mathrm f}\}=(\sigma_{\mathrm{SI}}^2/\sigma^2)\lambda_t\beta_{\Sigma2}^{\mathrm T}$.
Thus, the absolute residual SI grows linearly with the amplification gain and the number of active elements.
Nonetheless, since the desired signal scales with the same gain
through the cascaded term, the residual-SI-to-signal ratio remains governed by $\sigma_{\mathrm{SI}}^2/P_S$ and stays negligible throughout the practical gain range, confirming the robustness of Assumption~\ref{Assumption:1}.
At very high gains ($\delta_i>\delta_{\mathrm{sat}}$), amplifier nonlinearity becomes the dominant impairment. We model each active element with a Rapp-type soft-limiter amplifier \cite{rapp1991effects}, whose amplitude compression yields an effective per-element gain
$\delta_{\mathrm{eff}}=\delta\big(1+(\delta/\delta_{\mathrm{sat}})^{2p}\big)^{-1/(2p)}$,
where $p$ is the smoothness factor.
The desired effective channel is accordingly evaluated with $\delta_{\mathrm{eff}}$, while the power lost to compression is radiated as incoherent distortion, modeled as $D_{\mathrm n}=\varrho(\delta-\delta_{\mathrm{eff}})\bar\gamma\sum_{i}\beta_i^{\mathrm R}|h_i|^2|g_i|^2$ and $D_{\mathrm f}=\varrho(\delta-\delta_{\mathrm{eff}})\bar\gamma\sum_{i}\beta_i^{\mathrm T}|h_i|^2|t_i|^2$ with a distortion-scaling constant $\varrho$. Both impairments are treated as additional interference, so the denominators in \eqref{eq:R_cl_SI} and \eqref{eq:R_l_SI} become $\kappa_\ell(1-\alpha_c)\bar\gamma+Y_\ell+D_\ell+1$ and $\kappa_\ell\alpha_{\ell'}\bar\gamma+Y_\ell+D_\ell+1$, respectively.
As illustrated in Fig.~\ref{Fig:non_linear}, for $\delta\le\delta_{\mathrm{sat}}$~dB, the OP and EC curves with residual SI and with the nonlinear amplifier coincide with the ideal performance behaviors of Fig.~\ref{Fig:OP_EC_Amplification}.
However, when $\delta_i>\delta_{\mathrm{sat}}$, the accumulated SI and distortion cause a noticeable degradation.
This result identifies the boundary between the operating regime where the idealized model of Fig.~\ref{Fig:OP_EC_Amplification} remains valid and the regime where residual SI and amplifier nonlinearity must be explicitly modeled.
A rigorous analytical treatment of residual SI and amplifier nonlinearity is left for future work.

\subsection{Resource Allocation Evaluation}

\subsubsection{Convergence Behavior of Algorithm~\ref{Algo:Overall_algorithm}}

\begin{figure}[!t]
\centering
\includegraphics[width=\linewidth]{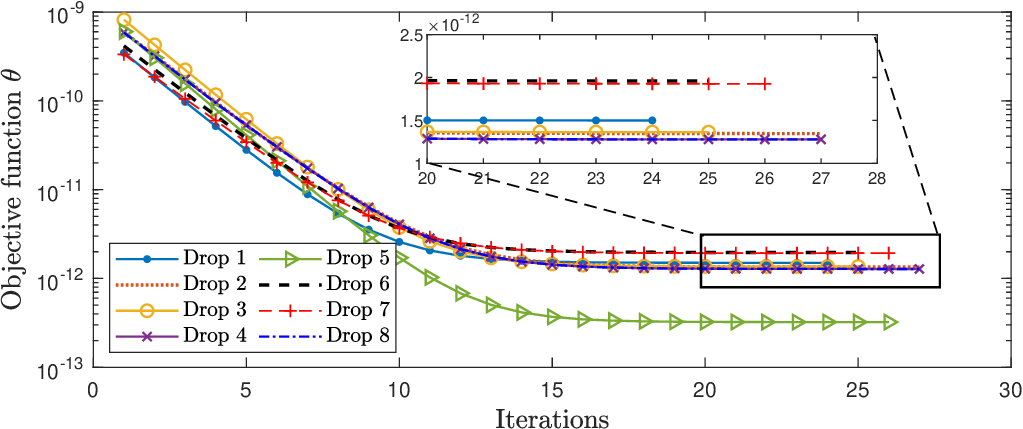}
\caption{Convergence behavior of the proposed Algorithm~\ref{Algo:Overall_algorithm} with eight arbitrary drops.}
\label{Fig:Convergence}
\end{figure}

Fig.~\ref{Fig:Convergence} depicts the convergence behavior of the proposed Algorithm~\ref{Algo:Overall_algorithm} across eight distinct drops.
Each curve shows the objective function $\theta$ as a function of the iteration index.
As discussed in Section~\ref{Sect:4.1.3}, the algorithm produces a non-increasing sequence of objective values that converges to a locally optimal solution.
The inset zooms into the convergence region, highlighting that all drops converge within fewer than 30 iterations.
This confirms the reliability and consistency of the proposed algorithm across multiple independent trials.

\subsubsection{Optimal Hyperparameters and FLOPs Comparison}

\begin{table}[!t]
\renewcommand{\arraystretch}{1.3}
\setlength{\tabcolsep}{3pt}
\centering
\caption{Optimal hyperparameters and FLOPs comparison among ML models.}
\label{Tab:HPO}
\footnotesize
\begin{tabular}{|c|p{2.5cm}|c|p{2.5cm}|c|}
\hline
\multirow{2}{*}{\textbf{Model}}
& \multicolumn{2}{c|}{\textbf{OP fairness problem}}
& \multicolumn{2}{c|}{\textbf{EC fairness problem}}\\
\cline{2-5}
& \centering\textbf{Optimal values}
& \textbf{FLOPs}
& \centering\textbf{Optimal values}
& \textbf{FLOPs}\\
\hline\hline
DMNN
& $L_D^{} = 3$, ${M_l} = 150$
& 138000
& $L_D^{} = 4$, ${M_l} = 125$
& 127500\\
\hline
\multirow{2}{*}{CMNN}
& ${L_C} = 2$, ${{\tilde L}_D} = 3$, ${{\tilde M}_l} = 150$, ${\cal K} = 32$, $H = 3$
& \multirow{2}{*}{206180}
& ${L_C} = 3$, ${{\tilde L}_D} = 3$, ${{\tilde M}_l} = 125$, ${\cal K} = 32$, $H = 3$
& \multirow{2}{*}{237620}\\
\hline
MXGB
& ${T} = 85 $, ${\cal D} = 50$
& 21675
& ${T} = 65 $, ${\cal D} = 60$
& 19825 \\
\hline
\end{tabular}
\end{table}

The optimal hyperparameters for all ML models are determined via HPO within predefined search spaces, shown in Table~\ref{Tab:Hyperparameters}.
The resulting optimal hyperparameters, along with the corresponding inference complexity, are summarized in Table~\ref{Tab:HPO}.
Since HPO is conducted over finite and practically motivated search spaces, the obtained configurations correspond to locally optimal solutions rather than guaranteed global optima. Nevertheless, the selected hyperparameter settings achieve a favorable tradeoff between prediction accuracy and computational efficiency, rendering them well suited for real-time RA inference.
Furthermore, Table~\ref{Tab:HPO} reveals that the MXGB model attains the lowest inference complexity, whereas the DMNN and CMNN models incur moderately higher FLOPs due to their deeper network architectures.

\subsubsection{Outage Probability Fairness Performance}

\begin{figure}[!t]
\centering
\includegraphics[width=\linewidth]{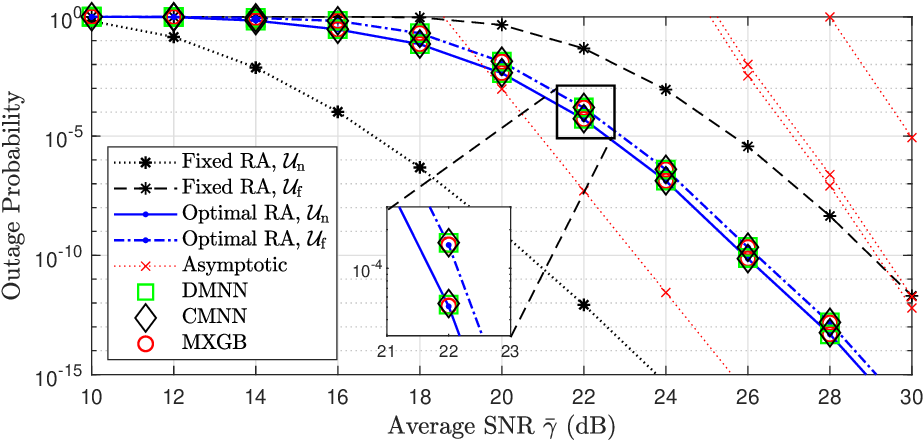}
\caption{OP fairness performance: Comparison between fixed and optimal RA configurations.}
\label{Fig:Optimal_OP}
\end{figure}

Fig.~\ref{Fig:Optimal_OP} compares the OP performance under fixed RA and the proposed optimal RA configuration.
Under fixed RA, a noticeable imbalance exists between the near user ${\cal U}_{\rm{n}}$ and the far user ${\cal U}_{\rm{f}}$,
where ${\cal U}_{\rm{f}}$ consistently experiences significantly higher OP due to its weaker effective channel. This imbalance becomes more pronounced in the moderate-SNR regime, highlighting the limitations of static PA and RS coefficients.
In contrast, the proposed optimal RA mechanism effectively reduces the OP gap between users by jointly optimizing the PA and RS coefficients based on the fairness-oriented formulation. Both users exhibit substantially improved OP performance, and the curves for ${\cal U}_{\rm{n}}$ and ${\cal U}_{\rm{f}}$ move much closer together.

\subsubsection{Ergodic Capacity Fairness Performance}

\begin{figure}[!t]
\centering
\includegraphics[width=\linewidth]{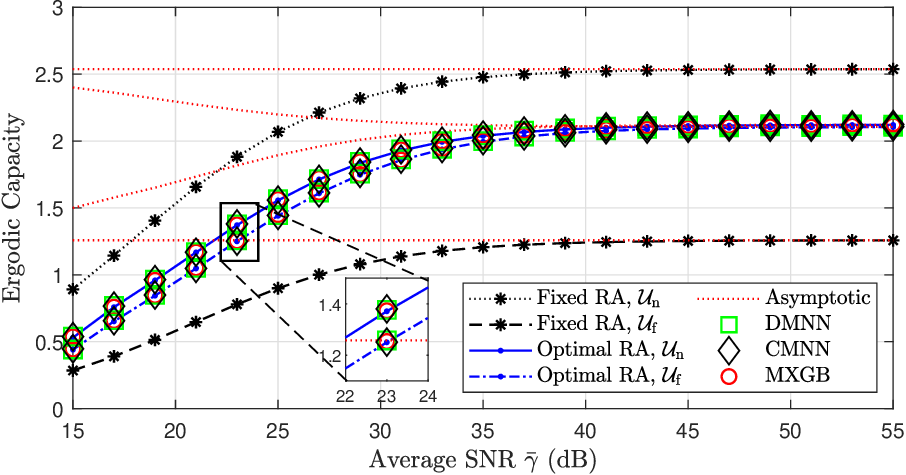}
\caption{EC fairness performance: Comparison between fixed and optimal RA configurations.}
\label{Fig:Optimal_EC}
\end{figure}

Fig.~\ref{Fig:Optimal_EC} illustrates the EC performance for fixed and optimal RA configurations.
With fixed RA, the EC of the near user ${\cal U}_{\rm{n}}$ is consistently higher than that of the far user ${\cal U}_{\rm{f}}$, reflecting the inherent channel disparity in STAR-RIS-assisted RSMA IoT systems.
This difference persists even as the SNR increases, indicating that fixed PA and RS coefficients cannot adequately balance the long-term capacity between users.
The optimal RA configuration significantly alleviates this imbalance. By dynamically adjusting the PA and RS factors,
the proposed fairness-oriented RA algorithm allows ${\cal U}_{\rm{f}}$ to achieve a higher EC while slightly reducing the EC of ${\cal U}_{\rm{n}}$.
This leads to a more equitable capacity distribution across users without compromising the system’s overall rate.

\subsubsection{ML-based RA Performance and Average Execution Time Comparison}

\begin{table}[!t]
\renewcommand{\arraystretch}{1.2}
\setlength{\tabcolsep}{2.8pt}
\centering
\caption{Average execution time comparison between the model-based and data-driven RA frameworks for various $N$ configurations.}
\label{table:exe_time}
\footnotesize
\begin{tabular}{|c|c|c|c|c|c|c|c|c|}
\hline
\multirow{2}{*}{$N$}
& \multicolumn{4}{c|}{\textbf{OP fairness problem}} & \multicolumn{4}{c|}{\textbf{EC fairness problem}} \\ \cline{2-9}
& \makecell{Alg.~1\\(s)} & \makecell{DMNN\\(ms)} & \makecell{CMNN\\(ms)} & \makecell{MXGB\\(ms)}
& \makecell{CVX\\(s)} & \makecell{DMNN\\(ms)} & \makecell{CMNN\\(ms)} & \makecell{MXGB\\(ms)} \\
\hline\hline
$5$   & 5.26  & 53.84 & 71.27 & 0.21 & 0.65 & 52.48 & 69.51 & 0.27 \\
$30$  & 8.21  & 53.47 & 72.61 & 0.46 & 0.94 & 50.29 & 64.86 & 0.14 \\
$100$ & 18.76 & 49.10 & 63.01 & 0.45 & 2.44 & 55.35 & 65.87 & 0.22 \\
$150$ & 28.23 & 52.13 & 68.83 & 0.17 & 3.65 & 49.89 & 65.66 & 0.39 \\
$250$ & 58.80 & 55.75 & 70.30 & 0.57 & 6.15 & 50.18 & 71.83 & 0.38 \\
\hline
\end{tabular}
\end{table}

Figs. \ref{Fig:Optimal_OP} and \ref{Fig:Optimal_EC} show that all ML-based frameworks achieve nearly identical fairness performance, as the variance between the predicted RA configurations and the ground-truth values is relatively small.\footnote{Notably, this behavior is expected because the considered learning problem is based on structured tabular data, for which the considered ML models are known to be particularly efficient and stable \cite{shwartz2022tabular}.} This low prediction variance of the RA configurations leads to weak sensitivity in the resulting fairness metrics.
It is worth noting that direct OP prediction is clearly inefficient due to its extremely small order of magnitude.
Moreover, when EC is predicted directly, its variance with respect to the ground-truth values can be significantly larger than that obtained by first predicting the RA configurations and subsequently computing EC.
Therefore, designing fairness-aware RA predictions in terms of OP and EC constitutes a more conservative and reliable approach.

Table~\ref{table:exe_time} reports the average execution time of the model-based approach and the data-driven frameworks for both the OP fairness and EC fairness problems under different $N$ configurations.
As $N$ increases, the execution time of the conventional optimization algorithm grows significantly, indicating poor scalability.
This trend reflects the inherent complexity of solving the model-based optimization problems, making such approaches increasingly impractical for large-scale or time-sensitive scenarios.
In contrast, the proposed ML-based methods exhibit substantially lower execution times, remaining in the millisecond range across all configurations.
Notably, once trained, these models demonstrate weak sensitivity to the number of STAR-RIS elements, indicating strong computational robustness.
Among them, MXGB consistently achieves the lowest inference time for both fairness problems, on the order of less than 1 ms.

The results in Fig.~\ref{Fig:Optimal_OP}, Fig.~\ref{Fig:Optimal_EC}, and Table~\ref{table:exe_time} suggest that the MXGB-based RA framework is the most efficient scheme for real-time fairness-aware RA strategies.

\subsubsection{Generalization and Extrapolation Robustness}\label{Sect:ML_Generalization}
The proposed learning models are trained on datasets generated over predefined parameter ranges that reflect practical deployment scenarios. A natural concern is how well these models generalize when the operating conditions at inference time differ from those seen during training. Following~\cite{sohrabi2021deep, tu2025hybrid}, the generalizability of the proposed models can be categorized into two groups.

The first category comprises parameters that affect only the input distribution while leaving the input and output dimensions of the models unchanged. In our framework, these correspond to the entries of the input feature vector ${\bf{x}}_s = [\bar\gamma, N, w_{\rm{n}},w_{\rm{f}}, \delta_i^{\rm{R}},\delta_i^{\rm{T}}, \beta_i^{\rm{R}},R_{\rm{c}}^{\rm{th}},R_{\rm{n}}^{\rm{th}},R_{\rm{f}}^{\rm{th}} ]^T$.
Since the effective channel-gain distribution $\kappa_\ell\sim\mathcal{GG}(\xi_\ell,1/2,\Omega_\ell^2)$ is governed by $\xi_\ell$ and $\Omega_\ell$, which are in turn functions of $N$, $\delta_i$, $\beta_i$, and $w_\ell$, sweeping these parameters over wide ranges during training exposes the models to a broad family of effective channel-gain statistics within the training support. To illustrate the extrapolation robustness of the trained models, we evaluate them at representative unseen operating points $\bar\gamma\in\{51,53,55\}$~dB, which lie $1$--$5$~dB outside the training range of $[0,50]$~dB. As observed in Fig.~\ref{Fig:Optimal_EC}, the DMNN-, CMNN-, and MXGB-based predictions still closely track the ground-truth optimal RA configurations at these extrapolated SNR values, confirming the extrapolation robustness of the proposed schemes along the SNR dimension. This behavior is consistent with practical learning-based deployments, where inference is typically performed on previously unseen data rather than on samples identical to the training set.
We emphasize, however, that this evaluation constitutes a mild, one-dimensional extrapolation of a single input feature beyond the training support, whereas the remaining features and the underlying channel model (i.e., i.i.d. Rayleigh fading with fixed large-scale statistics $\lambda_x$) are identical to those used in training. It is therefore not a genuine distribution-shift (out-of-distribution) test, in which the joint input statistics or the channel model at inference time differ from those of the training data. Accordingly, the robustness claimed here is limited to such moderate extrapolation, and a systematic evaluation under genuine distribution shifts is left for future work (Section~\ref{Sect:FutureWork}).

The second category comprises parameters that alter not only the input distribution but also the input and output dimensions of the models and the models' layers. The most representative example is the number of served users, since increasing this number enlarges the RA output vector of per-user PA and RS coefficients and requires per-user or per-group input features, which necessitates modifications to the network architecture and specialized training strategies. A related case is a change in deployment geometry that shifts the large-scale channel statistics $\lambda_x$; these are held fixed in the present feature set and would need to be incorporated as additional input features to generalize across arbitrary node distributions. Adapting the models to such dimension- and geometry-varying scenarios lies beyond the scope of this work and is left as a future research direction, in line with the multi-user scalability discussion in Section~\ref{Sect:Conclusion}. Under severe distribution shifts, lightweight transfer learning or fine-tuning of the pretrained models can further serve as a practical remedy.

\section{Concluding Remarks and Open Challenges}
\label{Sect:Conclusion}

\subsection{Concluding Remarks}
This work presented a comprehensive study of active STAR-RIS-assisted RSMA IoT systems, addressing the challenges posed by double-fading cascaded channels and the inherent imbalance between near and far users.
An integrated information-theoretic framework was developed to characterize system performance in terms of OP, EC, asymptotic OP, asymptotic EC, diversity order, and array gain, under both the presence and absence of direct source-destination links.
Moreover, both throughput-based and spectral-based EE metrics were analyzed, revealing the tradeoffs between transmission performance and energy consumption.

Building on this analysis, a fairness-oriented RA framework was proposed to jointly optimize PA and RS coefficients.
By leveraging SCA- and BCD-based convex reformulations, efficient algorithms were developed to ensure fairness in both outage and capacity metrics while satisfying system constraints.
Monte Carlo simulations consistently validated the analytical findings and demonstrated the effectiveness of active STAR-RIS architectures in mitigating the double-fading effect and improving RSMA performance across a broad range of operating conditions. Numerical results further highlighted the significant gains in fairness-enabled performance achieved by the proposed optimization framework.

However, the model-based RA approaches suffer from high execution time and limited scalability with respect to the number of STAR-RIS elements.
To address this limitation, a data-driven RA framework was introduced using DMNN, CMNN, and MXGB models. The results showed that these learning-based schemes can reliably approximate the ground-truth-based solutions with substantially reduced computational complexity. Among them, the proposed MXGB model consistently achieved the shortest execution time while maintaining performance comparable to optimization-based benchmarks, rendering it well suited for real-time fairness-aware RA strategies.

\subsection{Open Challenges and Future Research Directions}\label{Sect:FutureWork}
This work suggests several promising directions for future research.
First, on the multi-user side, the grouping/pairing extension of Remarks~\ref{Remark:MultiUser} and \ref{Remark:CapacityFairness_K} handles massive connectivity with an ${\cal O}(K\log K)$ SIC-ordering heuristic, whereas a full single-cell $K$-user treatment, requiring re-derived multi-user expressions with joint clustering and SIC-ordering optimization, remains open.
Second, on the learning side, beyond the SNR-extrapolation robustness verified in Section~\ref{Sect:ML_Generalization}, a systematic evaluation under genuine distribution shifts (e.g., Rician or spatially correlated fading, imperfect-CSI-corrupted inputs, or shifted large-scale statistics) remains open; moreover, generalizing the models across dimension- and geometry-varying settings (e.g., a varying number of users or shifting large-scale statistics) requires feature-set augmentation or architectural changes (Section~\ref{Sect:ML_Generalization}), with transfer learning or fine-tuning as a remedy under severe distribution shifts; along the same line, GNN-based RA is a promising route to topology-varying massive multi-user deployments.
Third, on the analytical side, a rigorous treatment of residual-SI accumulation and nonlinear amplification under high gains ($\delta>\delta_{\rm sat}$, complementing Section~\ref{Sect:Residual_SI}) and an iCSI-robust analytical and RA framework that re-derives the OP/EC expressions with channel-estimation error (beyond the numerical study in Section~\ref{Sect:iCSI}) are both valuable directions.
Fourth, on the deployment side, since the reported millisecond figures reflect algorithmic latency in simulation, the over-the-air, software-defined-radio, or hardware-emulation validation together with a realistic characterization of the STAR-RIS control and reconfiguration latency is essential to confirm end-to-end real-time operation.
Finally, an EE-aware scheme that jointly optimizes the PA, RS, and STAR-RIS amplification/ES coefficients under a Dinkelbach-based fractional program, thereby maximizing throughput- or spectral-based EE under the QoS and fairness constraints, is a promising extension of the present fairness-oriented design.

\bibliographystyle{IEEEtran}
\bibliography{reference}

\end{document}